\documentclass[runningheads]{llncs}
\usepackage{graphicx}
\usepackage{xspace}
\usepackage{thmtools} 
\usepackage{thm-restate}
\usepackage{todonotes}
\usepackage{enumerate}
\usepackage{soul}
\usepackage{hyperref}
\usepackage{xcolor}
\usepackage{mathtools}
\usepackage{thm-restate}

\renewcommand{\orcidID}[1]{\href{https://orcid.org/#1}{\includegraphics[scale=.03]{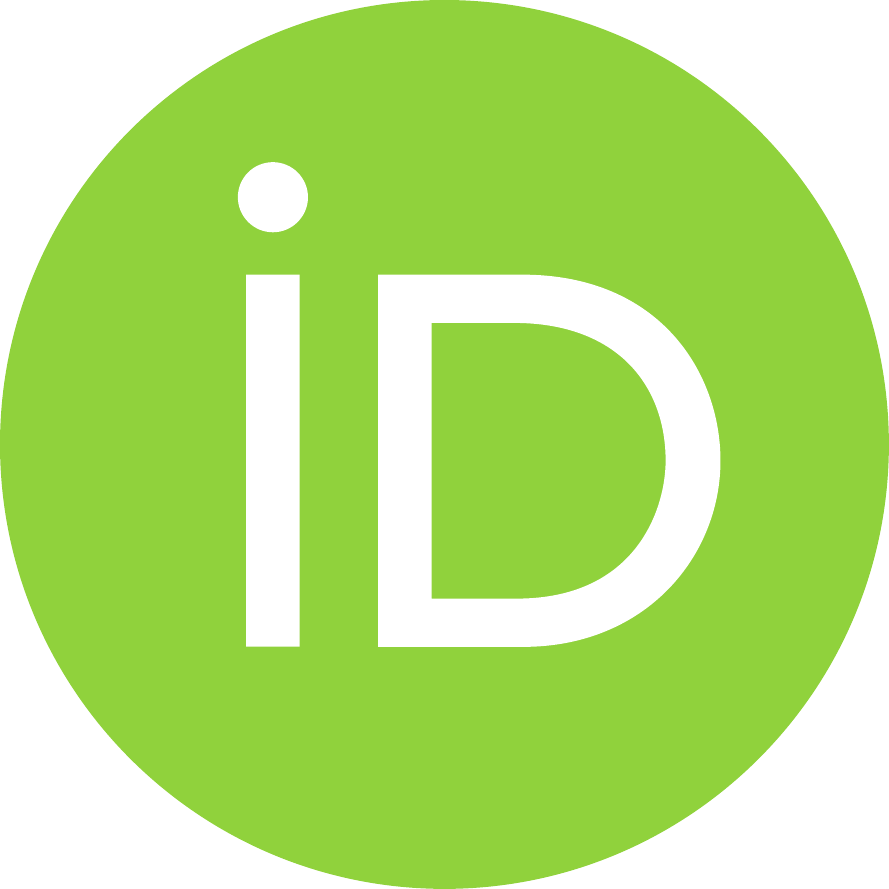}}} 

\usepackage[capitalise,nameinlink]{cleveref}
\Crefname{observation}{Observation}{Observations}
\Crefname{proposition}{Proposition}{Propositions}
\Crefname{claim}{Claim}{Claims}
\Crefname{property}{Property}{Properties}

\Crefname{enumi}{Property}{Properties}
\usepackage{paralist}
\usepackage{subfigure}
\usepackage{amssymb}
\usepackage{cite} 

\title{On the Complexity of the Storyplan Problem\thanks{Research partially supported by MIUR grant 20174LF3T8 {\em ``AHeAD: efficient Algorithms for HArnessing networked Data''}, Progetto RICBA21LG {\em "Algoritmi, modelli e sistemi per la rappresentazione visuale di reti"}, and the Vienna Science and Technology Fund (WWTF) grant ICT19-035.}}

\titlerunning{The Storyplan Problem}
\author{Carla Binucci\inst{1}\orcidID{0000-0002-5320-9110} \and Emilio Di~Giacomo\inst{1}\orcidID{0000-0002-9794-1928} \and William J. Lenhart\inst{2}\orcidID{0000-0002-8618-2444} \and \\Giuseppe Liotta\inst{1}\orcidID{0000-0002-2886-9694} \and Fabrizio Montecchiani\inst{1}\orcidID{0000-0002-0543-8912} \and Martin N\"ollenburg\inst{3}\orcidID{0000-0003-0454-3937}  \and Antonios Symvonis\inst{4}\orcidID{0000-0002-0280-741X}}

\authorrunning{C. Binucci et al.}

\institute{Department of Engineering, University of Perugia, Italy\\
	\email{name.surname@unipg.it}
	\and Department of Computer Science, Williams College, USA\\
	\email{wlenhart@williams.edu}
	\and Algorithms and Complexity Group, TU Wien, Vienna, Austria\\
	\email{noellenburg@ac.tuwien.ac.at }
	\and School of Applied Mathematical \& Physical Sciences, NTUA, Greece\\
	\email{symvonis@math.ntua.gr}
}

\newcommand{\story}{storyplan\xspace}
\newcommand{\storyp}{\textsc{StoryPlan}\xspace}
\newcommand{\storypfo}{\textsc{StoryPlanFixedOrder}\xspace}

\newcommand{\fw}{\text{\normalfont fw}}
\newcommand{\pw}{\text{\normalfont pw}}
\newcommand{\vc}{\kappa}
\newcommand{\fes}{\psi}
\newcommand{\ethlong}{Exponential Time Hypothesis\xspace}
\newcommand{\ethshort}{\textsc{ETH}\xspace}
\newcommand{\NP}{\textsf{NP}\xspace}
\newcommand{\XP}{\textsf{XP}\xspace}
\newcommand{\FPT}{\textsf{FPT}\xspace}
\newcommand{\sat}{\textsc{3SAT}\xspace}
\newcommand{\onesat}{\textsc{One-In-Three 3SAT}\xspace}
\newcommand{\sefe}{\textsc{SEFE}\xspace}
\newcommand{\ssefe}{\textsc{Sunflower SEFE}\xspace}

\begin{document}

\maketitle

\begin{abstract}
Motivated by dynamic graph visualization, we study the problem of representing a graph $G$ in the form of a \emph{\story}, that is, a sequence of frames with the following properties. Each frame is a planar drawing of the subgraph of $G$ induced by a suitably defined subset of its vertices. Between two consecutive frames, a new vertex appears while some other vertices may disappear, namely those whose incident edges have already been drawn in at least one frame. In a \story, each vertex appears and disappears exactly once. For a vertex (edge) visible in a sequence of consecutive frames, the point (curve) representing it does not change throughout the sequence. 

Note that the order in which the vertices of $G$ appear in the sequence of frames is a total order. 
In the \storyp problem, we are given a graph and we want to decide whether there exists a total order of its vertices for which a \story exists. We prove that the problem is \NP-complete, and complement this hardness with two parameterized algorithms, one in the vertex cover number and one in the feedback edge set number of $G$. Also, we prove that partial $3$-trees always admit a \story, which can be computed in linear time. Finally, we show that the problem remains \NP-complete in the case in which the total order of the vertices is given as part of the input and we have to choose how to draw the frames. 
\end{abstract}

\section{Introduction}
Let $G=(V,E)$ be a graph with $n$ vertices. We write $[n]$ as shorthand for the set $\{1,2,\dots,n\}$. A \emph{\story} $\mathcal{S}=\langle \tau, \{D_i\}_{i\in [n]} \rangle$ of $G$ is a pair defined as follows. The first element is a bijection $\tau: V \rightarrow [n]$ that represents a total order of the vertices of $G$. For a vertex $v \in V$, let $i_v= \tau(v)$ and let $j_v = \max_{u \in N[v]} \tau(u)$, where $N[v]$ is the set containing $v$ and its neighbors. The \emph{lifespan} of $v$ is the interval $[i_v,j_v]$. We say that $v$ \emph{appears} at step $i_v$, 
 is \emph{visible} at \emph{step} $i$ for each $i \in [i_v,j_v]$, and \emph{disappears} at step $j_v+1$. Note that a vertex does not disappear until all its neighbors have appeared. The second element of $\mathcal{S}$ is a sequence of drawings $\{D_i\}_{i\in [n]}$, such that: (i) each drawing $D_i$ contains all vertices  visible at step $i$, (ii) each drawing $D_i$ is planar, (iii) the point representing a vertex $v$ is the same over all drawings that contain $v$ (i.e., it does not change during the lifespan of $v$), and (iv) the curve representing an edge $e$ is the same over all drawings that contain $e$.  We introduce the \storyp problem.

\medskip\noindent\fbox{%
  \parbox{0.95\textwidth}{
    \storyp\\
    \textbf{Input:} Graph $G=(V,E)$\\
    \textbf{Question:} Does $G$ admit a \story?
  }%
}

\medskip\noindent In what follows, each drawing $D_i$ of a \story $\mathcal{S}$  is called a \emph{frame} of $\mathcal{S}$. Also, we denote by $|D_i|$  the number of vertices of $D_i$, while the \emph{width} of $\mathcal{S}$ is $w(\mathcal{S})=\max_{i \in [n]}|D_i|-1$ (we subtract one to align the definition with other width parameters).  If $G$  admits a \story, then the \emph{framewidth} of $G$, denoted by $\fw(G)$, is the minimum width over all its \story{s}; otherwise the framewidth of $G$ is conventionally set to $+\infty$. We will observe that the framewidth of $G$ upper bounds its pathwidth~\cite{DBLP:journals/jct/RobertsonS83}, since each frame can be interpreted as a bag of a path decomposition with the addition of conditions (ii)--(iv).

\medskip\noindent\textbf{Motivation and related work.} Testing for the existence of a \story of a graph generalizes planarity and it is of theoretical interest as it combines classical width parameters of graphs with topological properties. From a more practical perspective, computing a \story (if any) of a graph $G$  is a natural way to gradually visualize $G$ in a story-like or small-multiples fashion, such that each single drawing is planar and the reader's mental map is preserved throughout the sequence of drawings (see, e.g.,~\cite{DBLP:journals/vlc/GiacomoDLMT14} for a similar approach). More in general, the problem of visualizing graphs that change over time has motivated a notable amount of literature in graph drawing and network visualization (see, e.g.,~\cite{DBLP:conf/vissym/0001B0W14,DBLP:journals/ipl/BinucciBBDGPPSZ12,DBLP:journals/jgaa/BorrazzoLBFP20,DALOZZO20191,DBLP:conf/gd/SkambathT16,DBLP:journals/tvcg/VehlowBW16}).  
While numerous dynamic graph visualization models have been proposed, two works are of particular interest for our research. The first one is the work by Borrazzo et al.~\cite{DBLP:journals/jgaa/BorrazzoLBFP20}, in which the following problem is introduced. A \emph{graph story} is formed by a graph $G$, a total order of its vertices $\tau$, and a positive integer $W$. The problem is to find a sequence of drawings $\{D_i\}_{i \in [n]}$ in which  each  $D_i$ contains all vertices $v$ such that $i-W < \tau(v) \le i$, and the position of a vertex is the same over all drawings it belongs to. Borrazzo et al. prove that any story of a path or a tree can be drawn on a $2W \times 2W$ and on an $(8W +1)\times(8W +1)$ grid, respectively, so that all the drawings of the story are straight-line and planar. Note that having a fixed window of size $W$ implies that at most $O(W \cdot n)$ edges of $G$ can be represented, in particular, any edge whose endpoints are at distance larger than $W$ in $\tau$ does not appear in any drawing. Having both a fixed order and a fixed lifespan are the key differences with our setting. In particular,  unconstrained lifespans allow us to find stories in which all edges are drawn in at least one step, while planarity still guarantees that even large frames are readable. Besides such differences in the models, our focus is on the complexity of the decision problem, rather than on area bounds for specific graph families. 
The second work is by Da Lozzo and Rutter~\cite{DALOZZO20191}, who introduce \emph{stream planarity}. Given a graph $G$, a total order $\tau$ of the \emph{edges} of $G$, and a positive integer $W$, stream planarity asks for a sequence of drawings $\{D_i\}_{i \in [n]}$ in which  each  $D_i$ contains all edges $e$ such that $i-W < \tau(e) \le i$, and the subdrawing of the vertices and edges shared by $D_i$ and $D_{i-1}$ is the same in both drawings. Da Lozzo and Rutter prove that there exists a constant value for $W$ for which the stream planarity problem is \NP-complete. They also study a variant where a backbone graph is given whose edges must stay in the drawing at each time step; for this variant they prove that the problem is \NP-complete for all $W \ge 2$ and can be solved in polynomial time when $W=1$ or when the backbone graph is biconnected. The difference of stream planarity with our problem, besides the fact that edges are streamed rather than vertices, is again having a fixed order and a fixed lifespan. 

\bigskip\noindent\textbf{Contribution.} The main results in this paper can be summarized as follows.

\begin{itemize}
\item We show that \storyp is \NP-complete (\cref{sse:hardness}). As we reduce from \onesat and we blow up the instance by a linear factor, it  follows that there is no algorithm that solves \storyp in $2^{o(n)}$ time unless \ethshort fails. On the other hand, such a lower bound can be complemented with a simple algorithm running in $2^{O(n \log n)}$ time.

\item Motivated by the above hardness, we study the parameterized complexity of \storyp  and  describe two fixed-parameter tractable algorithms. We first show that \storyp belongs to \FPT when parameterized by the vertex cover number  via the existence of a kernel, whose size is however super-polynomial (\cref{sse:vc}). We then prove that \storyp parameterized by the feedback edge set number (i.e., the minimum number of edges whose removal makes the graph acyclic) admits a kernel of linear size (\cref{sse:fes}). 

\item  In parameterized analysis, a central parameter to consider is treewidth. In this direction, finding a parameterized algorithm for \storyp appears to be an elusive task. However, we show that for partial $3$-trees, a \story always exists and can be computed in linear time (\cref{sse:threetrees}).

\item Finally, we initiate the study of the complexity of a variant of \storyp in which the total order of the vertices is fixed in advance (but the vertex lifespan remains unconstrained).  We prove \NP-completeness for this problem via a reduction from \ssefe~\cite{DBLP:journals/jgaa/Schaefer13} (\cref{se:storypfo}). 
\end{itemize}


\noindent Some proofs are in the appendix and the corresponding statements are marked~($\star$).

\section{Preliminaries and Basic Results}\label{se:preliminaries}
A \emph{drawing} $\Gamma$ of a graph $G=(V,E)$ is a mapping of the vertices of $V$ to points in the plane $\mathbb{R}^2$, and of the edges of $E$ to Jordan arcs connecting their corresponding endpoints but not passing through any other vertex.
Drawing $\Gamma$ is \emph{planar} if no edge is crossed. A graph is \emph{planar} if it admits a planar drawing. 
A planar drawing of a planar graph $G$ subdivides the plane into topologically connected regions, called \emph{faces}. The infinite region is the \emph{outer face}. A \emph{planar embedding} $\mathcal E$ of $G$ is an equivalence class of planar drawings that define the same set of faces and the same outer face. 
For any $V' \subseteq V$, we denote by $G[V']$ the subgraph of $G$ induced by the vertices of $V'$ and by $\Gamma[V']$ the subdrawing of $\Gamma$ representing $G[V']$.

\bigskip\noindent\textbf{Connection with pathwidth.} The next properties show some simple connections between \story{s} and path decompositions~\cite{DBLP:journals/jct/RobertsonS83}. 

\begin{restatable}[$\star$]{theorem}{thpathwidth}\label{th:pathwidth}
Let $G=(V,E)$ be a graph, then $\pw(G) \le \fw(G)$. 
Also, if $G$ is planar then it always admits a \story, and in particular $\pw(G)=\fw(G)$.
\end{restatable}

\noindent Since computing the pathwidth is \NP-hard already for planar graphs of bounded degree~\cite{DBLP:journals/tcs/MonienS88}, the next corollary immediately follows from \cref{th:pathwidth}.

\begin{corollary}
Computing the framewidth of a graph is \NP-hard for planar graphs of bounded degree.
\end{corollary}

\noindent Analogously, computing the pathwidth of a graph is \FPT in the pathwidth~\cite{DBLP:journals/jal/BodlaenderK96}, hence computing the framewidth of a planar graph is also \FPT in the framewidth.

\bigskip\noindent\textbf{Complete bipartite graphs.} 
%
%
%
It is not difficult to verify that if a graph admits a \story, then it does not contain $K_5$ as a subgraph. 
However, complete bipartite graphs always admit a \story and such \story{s} have important properties. The next statement plays a central role in most of our proofs.

\begin{restatable}[$\star$]{lemma}{lebipartitethree}\label{le:bipartite-3}
Let $K_{a,b} = (A \cup B, E)$ be a complete bipartite graph with $a = |A|$,  $b=|B|$, and $3 \leq b \leq a$.
Let $\mathcal{S}=\langle \tau, \{D_i\}_{i\in [a+b]} \rangle$ be a \story of $K_{a,b}$. 
Exactly one of $A$ and $B$ is such that all its vertices are visible at some $i \in [a+b]$.
\end{restatable}

In view of \cref{le:bipartite-3}, we have the following definition.

\begin{definition}\label{def:fixflex}
For a complete bipartite graph $K_{a,b}$ with $3 \leq b \leq a$ and a \story $\mathcal{S}$ of $K_{a,b}$, we call \emph{fixed} the partite set of $K_{a,b}$ whose vertices are all visible at some step of $\mathcal{S}$, and \emph{flexible} the other partite set.
\end{definition}

\section{Complexity of \storyp}\label{se:storyp}

In this section we prove that: \storyp is \NP-complete and cannot be solved in $2^{o(n)}$ time unless \ethshort fails, but there is an algorithm running in $2^{O(n \log n)}$ time (\cref{sse:hardness}); \storyp is in \FPT parameterized by vertex cover number or feedback edge set number (\cref{sse:vc,sse:fes}); graphs of treewidth at most $3$ always admit a \story, which can be computed in linear time (\cref{sse:threetrees}). 

\subsection{Hardness}\label{sse:hardness}

\begin{figure}[t]
\centering
\subfigure[$K(x_i)$]{\label{fi:hard-a}\includegraphics[scale=1,page=2]{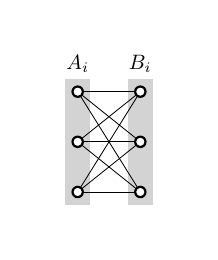}}\hfill
\subfigure[$K(C_i)$]{\label{fi:hard-b}\includegraphics[scale=1,page=3]{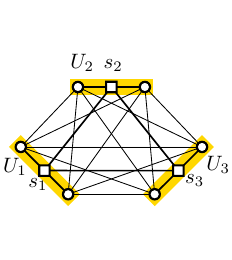}}\hfill
\subfigure[$W(l_{ij})$]{\label{fi:hard-c}\includegraphics[scale=1,page=4]{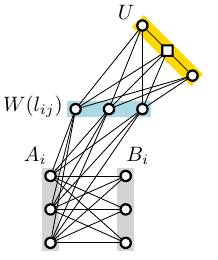}}
\caption{Illustration for the reduction of \cref{th:hardness}.\label{fi:hard}}
\end{figure}

We reduce from \onesat, a variant of 3SAT which asks whether there is a satisfying assignment in which \emph{exactly} one literal in each clause is true. Let $\varphi$ be a 3SAT formula over $N$ variables $\{x_i\}_{i \in [N]}$ and $M$ clauses   $\{C_i\}_{i \in [M]}$. We construct an instance of \storyp, i.e., a graph $G=(V,E)$, as follows; refer to \cref{fi:hard} for an illustration.

\smallskip\noindent\textbf{Variable gadget.} Each variable $x_i$ is represented in $G$ by a copy $K(x_i)$ of $K_{3,3}$ (see \cref{fi:hard-a}). Let $A_i$ and $B_i$ be the two partite sets of $K(x_i)$, which we call the \emph{v-sides} of $K(x_i)$. A true (false) assignment of $x_i$ will correspond to set $A_i$ being flexible (fixed) in a putative \story of $G$ (see \cref{def:fixflex}). 

\smallskip\noindent\textbf{Clause gadget.} Consider a copy of $K_{2,2,2}=(U_1 \cup  U_2 \cup U_3, F)$. An \emph{extended} $K_{2,2,2}$ is the graph obtained from any such a copy by adding three vertices $s_1,s_2,s_3$, such that these three vertices are pairwise adjacent, and each $s_j$ is adjacent to both vertices in $U_j$, for $j \in \{1,2,3\}$. In the following, $s_1,s_2,s_3$ are the \emph{special vertices} of the extended $K_{2,2,2}$, while the other vertices are the \emph{simple vertices}. A clause $C_i$ is represented in $G$ by an extended $K_{2,2,2}$, denoted by $K(C_i)$ (see \cref{fi:hard-b}). In particular, we call each of the three sets of vertices $U_j \cup \{s_j\}$ a \emph{c-side} of  $K(C_i)$. The idea is that $K(C_i)$ admits a \story if and only if exactly one c-side is flexible (each c-side will be part of a $K_{3,3}$, see the wire gadget below). 

\smallskip\noindent\textbf{Wire gadget.} Refer to \cref{fi:hard-c}. Let $x_i$ be a variable having a literal $l_{ij}$ in a clause $C_j$. Any such variable-clause incidence is represented in $G$ by a set  of three vertices, which we call the \emph{w-side} $W(l_{ij})$. All vertices of $W(l_{ij})$ are connected to all vertices of one of the three c-sides of $K(C_j)$, which we call $U$, such that the graph induced by $W(l_{ij}) \cup U$ in $G$ contains a copy of $K_{3,3}$. Also, each vertex of $W(l_{ij})$ is connected to all vertices of the v-side $A_i$ ($B_i$) if the literal is positive (negative), such that the graph induced by $W(l_{ij}) \cup A$ ($W(l_{ij}) \cup B$) in $G$ is a copy of $K_{3,3}$.  Also, note that each c-side of $K(C_j)$ is adjacent to exactly one w-side. 

\begin{restatable}[$\star$]{lemma}{lecorrectnessone}\label{le:correctness-1}
If graph $G$ admits a \story then $\varphi$ admits a satisfying assignment with exactly one true literal in each clause.
\end{restatable}
\begin{proof}[Sketch]
Let $\mathcal{S}$ be a \story of $G$. For each variable gadget $K(x_i)$ we assign the  value \emph{true} to $x_i$ if the v-side $A_i$  is flexible in $\mathcal{S}$. Consider any literal $l_{ij}$ and the wire gadget $W_{ij}$. If $l_{ij}$ is  positive (negative), then $A_i$  ($B_i$) and $W_{ij}$ form a $K_{3,3}$, hence by \cref{le:bipartite-3} the w-side $W_{ij}$ is fixed (flexible). Analogously, if we consider the clause gadget $K(C_j)$, the c-side connected with $W_{ij}$ is flexible (fixed). Symmetrically, we assign the value \emph{false} to $x_i$ if the v-side $B_i$ is instead flexible in $\mathcal{S}$, and for any positive (negative) literal $l_{ij}$, the w-side $W_{ij}$ is flexible (fixed), while the corresponding c-side of $K(C_j)$ is fixed (flexible).  In other words, the value of $x_i$ propagates consistently throughout all its literals. It remains to prove that, for any clause $C_j$ of $\varphi$, precisely one literal is true. Namely, we claim that exactly one c-side of $K(C_j)$ is flexible, while the other two are fixed. At high level, we rely on the fact that an extended $K_{2,2,2}$ wants at least two c-sides to be fixed, while the special vertices force at least one c-side to be flexible.
\end{proof}

\begin{restatable}[$\star$]{lemma}{lecorrectnesstwo}\label{le:correctness-2}
If the formula $\varphi$ admits a satisfying assignment with exactly one true literal in each clause, then graph $G$ admits a \story.
\end{restatable}
\begin{proof}[Sketch]
Given a satisfying assignment of $\varphi$ with one true literal per clause, we can compute a \story $\mathcal{S}=\langle \tau, \{D_i\}_{i \in [n]} \rangle$ of $G$. In what follows, when the order of a group of vertices is not specified, any relative order is valid.

\begin{figure}
\centering
\includegraphics[scale=1,page=11]{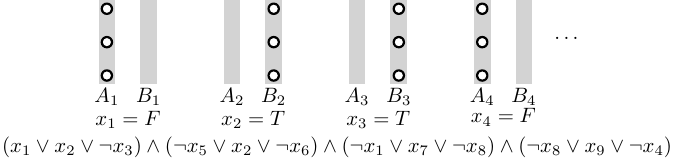}
\caption{Proof of \cref{le:correctness-2}: drawing the vertices of the fixed v-sides.\label{fi:draw-variable}}
\end{figure}

Consider a single variable gadget $K(x_i)$. If $x_i$ is true in the satisfying assignment, then we let appear the three vertices of the v-side $B_i$ of $K(x_i)$, that is, $B_i$ is the fixed side of $K(x_i)$. If $x_i$ is false, we do the opposite, namely we let appear the three vertices of the v-side $A_i$ of $K(x_i)$. This procedure is repeated for all variables in any order.  For ease of presentation, we can imagine that all the drawn v-sides are horizontally aligned, as shown in \cref{fi:draw-variable}. Thus, for the variable gadgets, it remains to draw their flexible v-sides.

\begin{figure}
\centering
\includegraphics[scale=1,page=12]{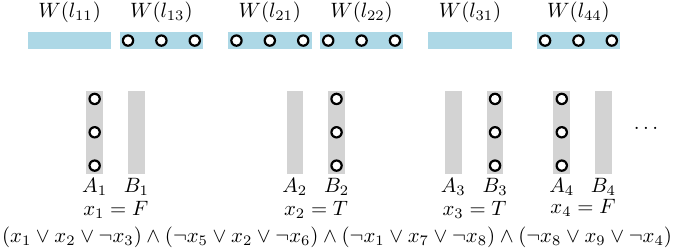}
\caption{Proof of \cref{le:correctness-2}: drawing the vertices of the fixed w-sides.\label{fi:draw-wires}}
\end{figure}

Consider now a wire gadget $W(l_{ij})$. If $x_i$ is true and $l_{ij}$ is positive, then $W(l_{ij})$ must be fixed because it forms a $K_{3,3}$ with the v-side  $A_i$  of $K(x_i)$, which is flexible. Therefore we let appear the three vertices of $W(l_{ij})$.  Similarly, if $x_i$ is false and $l_{ij}$ is negative,  then $W(l_{ij})$ must be fixed, and we let appear the three vertices of $W(l_{ij})$. Again, this procedure is repeated for all wires in any order. For ease of presentation, we can imagine that all the drawn w-sides are arranged along a horizontal line slightly above the variable gadgets, as shown in \cref{fi:draw-wires}. Thus, also for the wire gadgets, it remains to draw the flexible w-sides.

\begin{figure}
\centering
\includegraphics[scale=1,page=13]{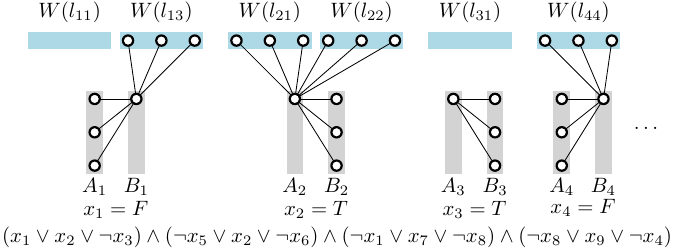}
\caption{Proof of \cref{le:correctness-2}: drawing the vertices of the flexible v-sides.\label{fi:draw-flex-variable}}
\end{figure}

\begin{figure}
\centering
\subfigure[]{\includegraphics[scale=0.5,page=6]{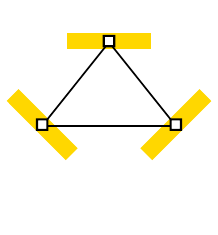}\label{fi:draw-clause-a}}\hfill
\subfigure[]{\includegraphics[scale=0.5,page=7]{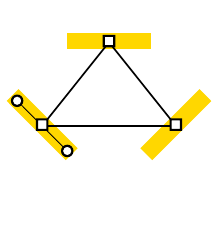}\label{fi:draw-clause-b}}\hfill
\subfigure[]{\includegraphics[scale=0.5,page=8]{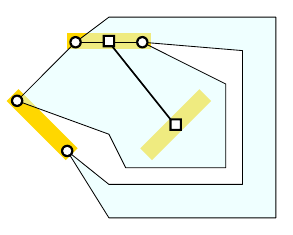}\label{fi:draw-clause-c}}\hfill
\subfigure[]{\includegraphics[scale=0.5,page=9]{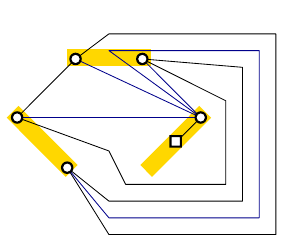}\label{fi:draw-clause-d}}\hfill
\subfigure[]{\includegraphics[scale=0.5,page=10]{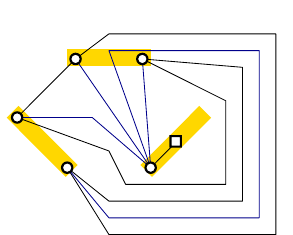}\label{fi:draw-clause-e}}\hfill
\caption{Proof of \cref{le:correctness-2}: drawing a clause gadget.\label{fi:draw-clause}}
\end{figure}

We sketch the remaining part of the proof (see \cref{ap:hardness}). Flexible v-sides can be drawn as in \cref{fi:draw-flex-variable}. \cref{fi:draw-clause} shows how to draw a clause gadget, ignoring the connections with the linked wire gadgets. Finally, in order to draw the flexible w-sides and their edges, and the edges between the fixed w-sides and the corresponding c-sides, we enclose all the wire and variable gadgets in a face of the current clause gadget where all vertices of the linked c-side are visible.
\end{proof}


\begin{restatable}[$\star$]{theorem}{thhardness}\label{th:hardness}
The \storyp problem is \NP-hard and it has no $2^{o(n)}$ time algorithm unless \ethshort fails.
\end{restatable}

The above lower bound for the running time of an algorithm solving \storyp can be easily complemented with a nearly tight upper bound. The proof of the next theorem also shows that \storyp belongs to \NP. Namely, it gives a nondeterministic scheme to generate a set of candidate solutions, and then it shows how to check, in polynomial time, if a candidate solution is valid. However, it intertwines such process in order to obtain a lower time complexity. 

\begin{restatable}[$\star$]{theorem}{thupper}\label{th:upper}
The \storyp problem is in \NP. Also, given an $n$-vertex graph $G$, there is an algorithm that solves \storyp on $G$ in $2^{O(n \log n)}$ time.
\end{restatable}
\begin{proof}[Sketch]
We first guess a total order of the vertices of $G$. This fixes for each $i \in [n]$ the visible vertices. Next, for each $i \in [n]$, we generate the possible planar embedddings (rather than planar drawings) of the graph induced by the vertices visible at step $i$, and discard any embedding $\mathcal E$ for which there is no planar embedding $\mathcal E'$ generated at step $i-1$ (if $i>1$) such that the restrictions of  $\mathcal E$ and  $\mathcal E'$ to the common subgraph coincide. If the algorithm returns at least one planar embedding at step $n$, there is a sequence of planar embeddings in which common subgraphs share the same embedding, hence $G$ admits a \story.
\end{proof}

\subsection{Parameterization by Vertex Cover Number}\label{sse:vc}

A \emph{vertex cover} of a graph $G=(V,E)$ is a set $C \subseteq V$ such that every edge of $E$ is incident to a vertex in $C$, and the \emph{vertex cover number} of $G$ is the minimum size of a vertex cover of $G$. We prove the following by means of kernelization.

\begin{restatable}[$\star$]{theorem}{thfptvc}\label{th:fpt-vc}
Let $G=(V,E)$ be a graph with $n$ vertices and vertex cover number $\vc=\vc(G)$.
Deciding whether $G$ admits a \story, and computing one if any, can be done in $O(2^{2^{O(\vc)}}+n^2)$ time.
\end{restatable}

\medskip\noindent\textbf{Algorithm Description.} 
Without loss of generality, we assume that the input graph $G$ does not contain isolated vertices, as they do not affect the existence of a \story. Let $C$ be a vertex cover of size $\vc=\vc(G)$ of graph $G$. For $U \subseteq C$, a vertex  $v \in V \setminus C$ is of \emph{type $U$} if $N(v) = U$, where $N(v)$ denotes the set of neighbors of $v$ in $G$. This defines an equivalence relation on $V \setminus C$ and in particular partitions $V \setminus C$ into at most $\sum_{i=1}^{\vc} {\vc \choose{i}}=2^{\vc}-1 < 2^\vc$ distinct types. Denote by $V_U$ the set of vertices of type $U$. We define three reduction rules.

\smallskip\noindent\textbf{R.1:} {\em If there exists a type $U$ such that $|U| = 1$, then pick an arbitrary vertex $x \in V_U$ and remove it from $G$.}

\smallskip\noindent\textbf{R.2:} {\em If there exists a type $U$ such that $|U| = 2$ and $|V_U| > 1$, then pick an arbitrary vertex $x \in V_U$ and remove it from $G$.}

\smallskip\noindent\textbf{R.3:} {\em If there exists a type $U$ such that $|U|  \ge 3$ and $|V_U| > 3$, then pick an arbitrary vertex $x \in V_U$ and remove it from $G$.}

\begin{restatable}[$\star$]{lemma}{lekernel}\label{le:kernel}
Let $G'$ be the graph obtained from $G$ by applying one of the reduction rules \textbf{R.1--R.3}. 
Then $G$ admits a \story if and only if $G'$ does. 
\end{restatable}
\begin{proof}[Sketch]
For the nontrivial direction, suppose that $G'$ admits a \story $\mathcal{S'}=\langle \tau', \{D'\}_{i\in [n']}\rangle$, where $n'=n-1$. 
\begin{figure}
\centering
\includegraphics[width=\textwidth,page=1]{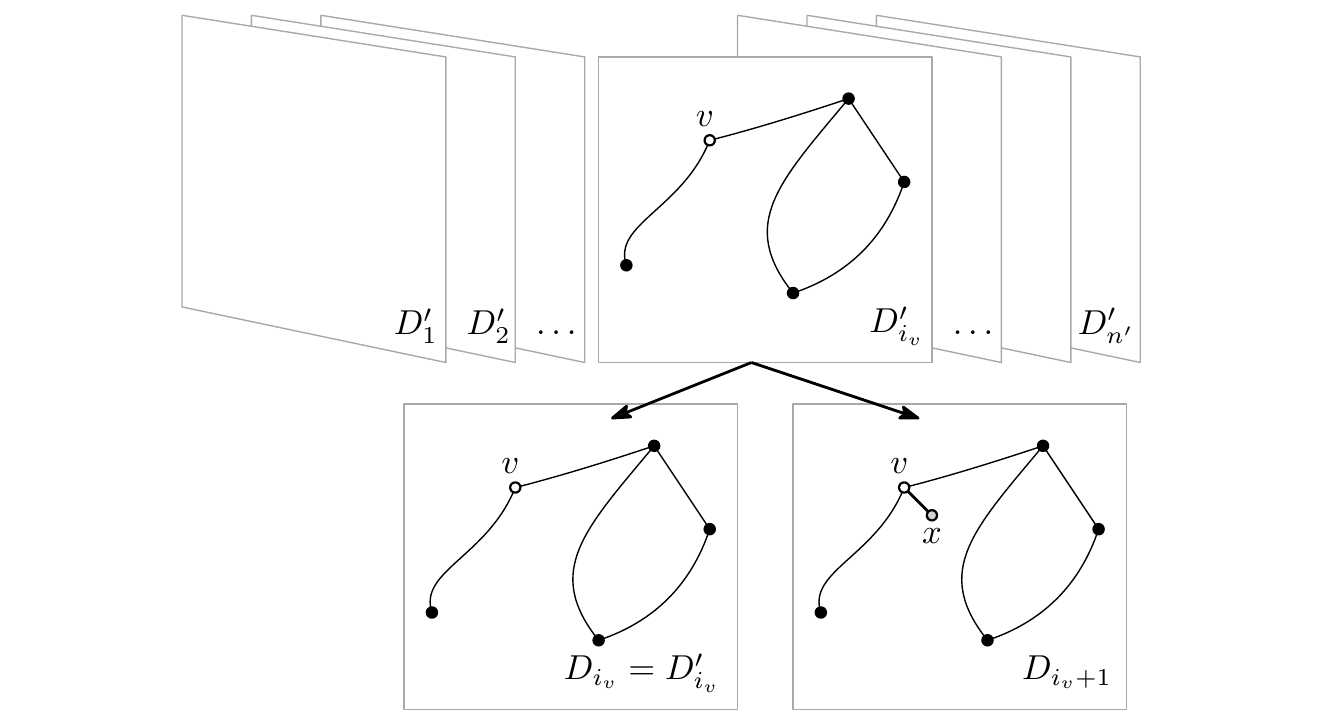}
\caption{Illustration for \textbf{Case A} of the proof of \cref{le:kernel}.\label{fi:fpt-1}}
\end{figure}
We can distinguish three cases based on the reduction rule applied to $G$. Here we only prove the simplest of the three cases, namely the case in which \textbf{R.1} is applied. See \cref{fi:fpt-1} for an illustration. Let $x$ be the vertex removed from $G$ to obtain $G'$ and let $v$ be its neighbor, whose lifespan according to $\tau'$ is $[i_v,j_v]$. We compute $\tau$ from $\tau'$ by inserting $x$ right after $v$, thus the lifespan of $x$ in $\tau$ is $[i_v+1,i_v+1]$. Similarly, we compute $\{D_i\}_{i\in [n]}$ from $\{D'\}_{i\in [n']}$ as follows. For each $i \le i_v$, we set $D_i = D'_i$. For $i=i_v+1$, we draw $x$ in $D'_{i_v}$ sufficiently close to $v$ such that edge $xv$ can be drawn as a straight-line segment that does not intersect any other edge. We then set $D_i$ to be equal to the resulting drawing. For each $i>i_v+1$, we set $D_i = D'_{i-1}$.
\end{proof}

Based on \cref{le:kernel}, we can construct an equivalent instance of  $G$ of size $O(2^{\vc})$ and use it to conclude the proof of \cref{th:fpt-vc} (see~\cref{apx:vcn}).

\subsection{Parameterization by Feedback Edge Set}\label{sse:fes}

A \emph{feedback edge set} of a graph $G=(V,E)$ is a set $F \subseteq E$ whose removal results in an acyclic graph, and the \emph{feedback edge set number} of $G$ is the minimum size of a feedback edge set of $G$. 
We prove the following.

\begin{theorem}\label{th:fpt-fes}
Let $G$ be a graph with $n$ vertices and feedback edge set of size $\fes=\fes(G)$.
Deciding whether $G$ admits a \story, and computing one if any, can be done in $O(2^{O(\fes \log \fes)}+n^2)$ time.
\end{theorem}

\medskip\noindent\textbf{Algorithm Description.} A \emph{$k$-chain} of $G$ is a path with $k+2$ vertices and such that its $k$ inner vertices all have degree two. We define two reduction rules.

\smallskip\noindent\textbf{R.A:} {\em If there exists a vertex of degree one, then remove it from $G$.}

\smallskip\noindent\textbf{R.B:} {\em If there exists a $k$-chain with $k \ge 3$, then remove its inner vertices from~$G$.}

\smallskip Based on the above reduction rules we can prove the following.

\begin{restatable}[$\star$]{lemma}{lekernelfessize}\label{le:kernel-fes-size}
\storyp parameterized by feedback edge set number admits a kernel of linear size.
\end{restatable}

\noindent To conclude the proof of \cref{th:fpt-fes}, observe that computing a linear kernel $G^*$ of $G$, i.e., applying exhaustively the reduction rules \textbf{R.A} and \textbf{R.B}, can be done in $O(n+\fes)$ time. Afterwards, following the lines of the proof of \cref{th:fpt-vc}, we can brute-force a solution for $G^*$ (if any) in $2^{O(\fes \log \fes)}$ time, and reinsert the missing $O(n)$ vertices each in $O(n)$ time (as detailed \cref{le:kernel-fes} of \cref{ap:fes}).

\subsection{Partial $3$-trees}\label{sse:threetrees}


A \emph{$k$-tree} has a recursive definition: A complete graph with $k$ vertices is a $k$-tree; for any $k$-tree $H$, the graph obtained from $H$ by adding a new vertex $v$ connected to a clique $C$ of $H$ of size $k$ is a $k$-tree; $C$ is the \emph{parent clique} of $v$. A \emph{partial $k$-tree} is a subgraph of a $k$-tree and partial $k$-trees are exactly the graphs of treewidth at most $k$. Since $2$-trees are planar, they admit a \story by \cref{th:pathwidth}. We prove that the same holds for partial $3$-trees (which may be not planar).

\newcommand{\T}{\mathcal T}

\begin{restatable}[$\star$]{theorem}{ththreetrees}\label{th:3trees}
Every partial $3$-tree $G$ with $n$ vertices admits a \story, which can be computed in $O(n)$ time.
\end{restatable}
\begin{proof}[Sketch]
	We shall assume that $G$ is a (non-partial) $3$-tree. Indeed, if $G$ is a partial $3$-tree, a supergraph of $G$ that is a  $3$-tree always exists by definition. 
	\begin{figure}[htbp]
		\centering
		\subfigure[]{\label{fi:3-trees-a}\includegraphics[width=0.4\linewidth,page=1]{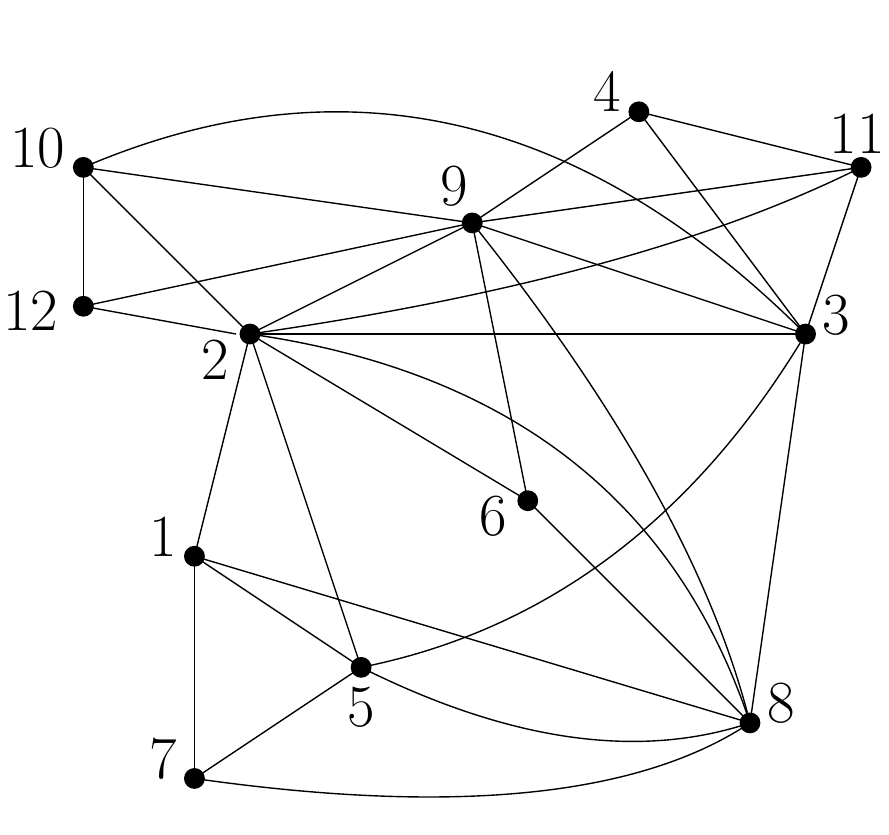}}\hfill
		\subfigure[]{\label{fi:3-trees-b}\includegraphics[width=0.5\linewidth,page=2]{figs/3-trees}}\\
		\subfigure[]{\label{fi:3-trees-c}\includegraphics[width=0.45\linewidth,page=3]{figs/3-trees}}\hfill
		\caption{(a) A $3$-tree $G$. (b) The decomposition tree $\T$ of $G$. (c) The subgraph $G_{\mu}$ of $G$ associated with  bag $\mu$, highlighted in (b). 
		\label{fi:3-trees}}
	\end{figure}
	We now construct a specific tree decomposition $\T$ of $G$ that will be used to compute its \story; refer to \cref{fi:3-trees}. For a definition of tree decomposition see~\cite{ROBERTSON1986309}. The subgraph $C_{\mu}$ induced by the vertices of each bag $\mu$ of $\T$ is the \emph{subgraph associated with $\mu$} and it is a $4$-clique for each bag $\mu$ of $T$, except for the root $\rho$ of $\T$ for which $C_{\rho}$ is the initial $3$-cycle.  The subgraph $C_{\mu}$ contains four $3$-cliques, three of them are \emph{active} (this means that they can appear in some subgraph $C_{\nu}$ associated with a child $\nu$ of $\mu$) and one is \emph{non-active}. The unique $3$-clique in $C_{\rho}$ is active. Each bag $\mu$ of $\T$ has one child $\nu$ for each vertex $v$ whose parent clique is an active $3$-clique of $\mu$.  The $4$-clique $C_\nu$ consists of the parent clique $C$ of $v$, vertex $v$, and the edges connecting $v$ to $C$; $C$ is non-active in $\nu$, while the other three $3$-cliques are active. For each bag $\mu$ distinct from $\rho$, we denote by $v_{\mu}$ the vertex shared by the three active $3$-cliques of $C_{\mu}$. We say that $v_{\mu}$ is \emph{associated with} $\mu$. One easily verifies that $\T$ is a tree decomposition. Also, $\T$ has $n-2$ bags: the root and a bag for each vertex of $G$ that is not in the initial $3$-cycle. 
	
	We now associate to each bag $\mu$ a subgraph $G_\mu$ of $G$. For the root $\rho$, the subgraph $G_{\rho}$ is the initial $3$-cycle. For a bag $\mu$ with parent $\lambda$, the subgraph $G_\mu$ is obtained from $G_{\lambda}$ by connecting $v_{\mu}$ to the vertices of its parent clique. 
	
	\begin{property}\label{pr:3-tree-one}
		Every graph $G_{\mu}$ is an embedded planar $3$-tree such that the active $3$-cliques of $C_{\mu}$ are internal faces of $G_{\mu}$.
	\end{property}

	The proof of \cref{pr:3-tree-one} is by induction on the length of the path from the root $\rho$ to $\mu$ in $\T$. The graph $G_{\rho}$ consists of a $3$-cycle, which is the unique active $3$-clique and which is both an internal and an external face. The graph $G_{\mu}$ is obtained by adding the vertex $v_{\mu}$ to $G_{\lambda}$ and connecting it to its parent clique $C$, which is active in $C_{\lambda}$. By induction, $C$ is an internal face of $G_{\lambda}$ and therefore, by placing $v_{\mu}$ inside this face, we obtain an embedded planar $3$-tree such that the active faces of $C_{\mu}$ are the three faces created by the addition of $v_{\mu}$ inside $C$.    
	
	Let $\mu \neq \rho$ be a bag of $\T$; the next property follows from the definition of $\T$.
	
	\begin{property}\label{pr:3-tree-two}
		The neighbors of $v_{\mu}$ distinct from those of its parent clique are all vertices associated with bags of the subtree of $\T$ rooted at $\mu$.
	\end{property}

	Let $\rho=\mu_1, \mu_2, \dots, \mu_{n-2}$ be an order of the bags of $\T$ according to a preorder visit of $\T$. To create a \story of $G$, we define an ordering $\tau: v_1, v_2, \dots, v_n$ of the vertices of $G$ such that $v_1$, $v_2$, and $v_3$ are the vertices of the initial $3$-cycle, and each $v_i$ with $i > 3$ is the vertex associated with $\mu_{i-2}$.
	
	
Let $G_i$ be the graph induced by the vertices that are visible at step $i$, for $i \geq 3$; by \cref{pr:3-tree-two} the graph $G_i$ is a subgraph of $G_{\mu_i}$ which, by \cref{pr:3-tree-one} is an embedded planar $3$-tree such that the three active $3$-cliques of $C_{\mu_i}$ are faces of $G_{\mu_i}$. To simplify the description we prove that there exists a \story $\mathcal{S}=\langle \tau, \{D_i\}_{i\in [n]} \rangle$, where each $D_i$ is a drawing of $G_{\mu_i}$. This implies that there exists a \story where each $D_i$ is a drawing of $G_i$. Let $\mu_j$ be the parent of $\mu_i$ in $\T$; since the order $\tau$ corresponds to a preorder of the bags of $\T$, we have $j < i$. Moreover, all bags $\mu_k$ with $j < k <i$, if any, belong to the subtrees of $\mu_j$ visited before $\mu_i$ and for each such subtree $\T'$ no other bag of $\T'$ exists before $\mu_j$ or after $\mu_i$. By \cref{pr:3-tree-two} all the vertices associated with the bags $\mu_k$ that belong to $G_{\mu_{i-1}}$ do not have any neighbor after $v_{\mu_i}$ and therefore they can be removed. The removal of these vertices transforms $G_{\mu_{i-1}}$ into $G_{\mu_j}$ (all the vertices associated with the bags $\mu_k$ for $j < k <i$ had been added to $G_{\mu_j}$ that had never been changed). By \cref{pr:3-tree-one} the active $3$-cliques  of $G_{\mu_j}$ are faces of $G_{\mu_j}$. It follows that there exists a \story $\mathcal{S}=\langle \tau, \{D_i\}_{i\in [n]} \rangle$ whose frames $D_i$ are as follows. $D_1$ is a planar drawing of a $3$-cycle; given  $D_{i-1}$ of $G_{\mu_{i-1}}$, a drawing $D_i$ of $G_{\mu_i}$ can be computed by removing all vertices associated with the bags $\mu_k$ for $j < k <i$, and adding $v_{\mu_i}$ inside a face of $D_j$.

	The above \story can be computed in $O(n)$ time, see \cref{ap:3tree}.
\end{proof}

\section{Complexity with Fixed Order}\label{se:storypfo}

In this section we study the \storypfo variant of \storyp, defined below. This variant is closer to the setting studied in~\cite{DBLP:journals/jgaa/BorrazzoLBFP20} and it models the case in which the way the graph changes over time is prescribed.  

\medskip\noindent\fbox{%
  \parbox{0.95\textwidth}{
    \storypfo\\
    \textbf{Input:} Graph $G=(V,E)$ with $n$ vertices, total order $\tau:V \rightarrow [n]$\\
    \textbf{Question:} Does $G$ admit a \story $\mathcal{S}=\langle \tau, \{D_i\}_{i\in [n]} \rangle$?
  }%
}

\medskip\noindent We prove that \storypfo is \NP-complete by reducing from \ssefe. Let $G_1,\dots,G_k$ be $k$ graphs on the same set $V$ of vertices. A \emph{simultaneous embedding with fixed edges (SEFE)} of $G_1,\dots,G_k$ consists of $k$ planar drawings $\Gamma_1,\dots,\Gamma_k$  of $G_1,\dots,G_k$, respectively, such that each vertex  is mapped to the same point in every drawing and each shared edge is represented by the same simple curve in all drawings sharing it. The \sefe problem asks whether $k$ input graphs on the same set of vertices admit a SEFE, and it is \NP-complete  even when the pairwise intersection between any two input graphs is the same over all pairs of graphs~\cite{ANGELINI201571,DBLP:journals/jgaa/Schaefer13}. This variant is called \ssefe, and the result in~\cite{ANGELINI201571,DBLP:journals/jgaa/Schaefer13} proves \NP-completeness already when $k=3$. 

\begin{figure}[t]
\centering
\subfigure[]{\includegraphics[page=1]{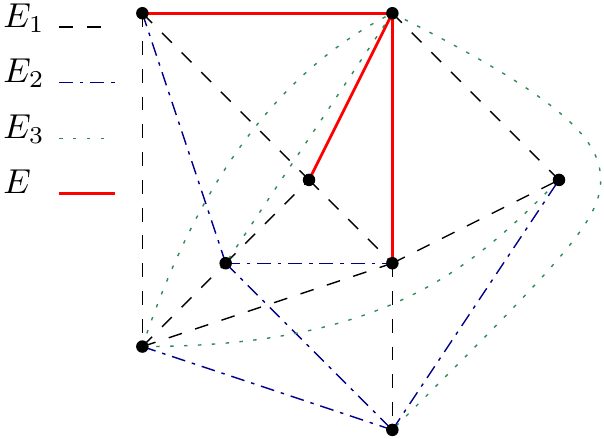}}\hfil
\subfigure[]{\includegraphics[page=2]{figs/fo-hard}}
\caption{Illustration for \cref{th:fohardness}. (a) An instance $G_1,G_2,G_3$ of \ssefe; 
	(b) The instance $G$ constructed from $G_1,G_2,G_3$; the subdivision vertices are circles with a white fill while the spectators are squares.\label{fi:fohard}}
\end{figure}

\medskip\noindent\textbf{Construction.} Refer to \cref{fi:fohard} for an example. Let $G_1,G_2,G_3$ be an instance of \ssefe. Let $V$ be the common vertex set of the three graphs, let $E$ be the common edge set, and let $E_i$ be the exclusive edge set of $G_i$ for $i=1,2,3$. We construct an instance $\langle G, \tau \rangle$ of \storypfo as follows. Graph $G$ contains all vertices in $V$ and all edges in $E$. Also, for each edge $e=uv$ in $E_i$, it contains a vertex $w^i_e$, called \emph{a subdivion vertex of $E_i$}, and the edges $uw^i_e$ and $vw^i_e$ (i.e., it contains the edge $e$ subdivided once). Moreover, for each vertex $z$, either a vertex in $V$ or a subdivision vertex of an edge, $G$ contains an additional vertex $s_z$, called the \emph{spectator of $z$}, and the edge $zs_z$. To obtain the total order $\tau$ we group the vertices of $G$ in a set of blocks $B_1,\dots,B_8$, and we order the blocks by increasing index, while vertices within the same block can be ordered arbitrarily. We denote by $\tau^-_i$ and $\tau^+_i$ the position in $\tau$ of the first and of the last vertex of $B_i$, respectively, for each $i=1,\dots,8$. Block $B_1$ contains all vertices in $V$; for $i \in \{2,4,6\}$, block $B_i$ contains all subdivision vertices of $E_{\frac{i}{2}}$, while block $B_{i+1}$ contains all spectators of the vertices in $B_i$; finally $B_8$ contains all spectators of the vertices in $B_1$.

\begin{restatable}[$\star$]{theorem}{thfohardness}\label{th:fohardness}
The \storypfo problem is \NP-complete.
\end{restatable}
\begin{proof}[Sketch]
At a high level, the total order $\tau$ is designed to show the three graphs one by one while keeping the common edge set visible. In particular, a spectator vertex $s_v$ forces vertex $v$ to stay visible until $s_v$ appears, while a subdivision vertex $w^i_e$ makes edge $e$ visible only when $G_i$ must be drawn.
\end{proof}

\section{Discussion and Open Problems}\label{se:open}

Our work can stimulate further research based on several possible directions.

\begin{itemize}

\item It would be interesting to study further parameterizations of \storyp. Is \storyp parameterized by treewidth (pathwidth) in \XP? In addition, we note that if the total order is fixed, then  an \FPT algorithm in the size of the largest frame (or the length of the longest lifespan) readily follows from the proof of \cref{th:upper}.

\item Conditions (iii) and (iv) of the definition of a \story can be replaced by the existence of a sequence of planar embeddings in which common subgraphs keep the same embedding. This is not true if we study more geometric versions of the problem, in which for instance edges are straight-line segments and/or vertices are restricted on an integer grid of fixed size (as in~\cite{DBLP:journals/jgaa/BorrazzoLBFP20}).

\item Condition (ii) of \story can be relaxed so to only allow specific crossing patterns~\cite{DBLP:journals/csur/DidimoLM19,DBLP:books/sp/20/HT2020}, e.g., right-angle crossings or few crossings per edge.

\end{itemize}

\bibliographystyle{splncs04}
\bibliography{bibliography}

\clearpage
\appendix
\section*{Appendix}

\section{Missing Proofs of \cref{se:preliminaries}}

A \emph{path decomposition} of a graph $G=(V,E)$ is a sequence $\{X_i\}_{i \in [h]}$ of subsets of $V$, called \emph{bags}, for some integer $h$, such that: (1) for each edge $e=uv \in E$, there exists a bag $X_i$ containing both $u$ and $v$, and (2) for every three indices $i \le j \le k$, $X_i \cap X_k \subseteq X_j$. The \emph{width} of $\{X_i\}_{i \in [h]}$ is equal to $\max_{i \in [h]} |X_i| -1$, and the \emph{pathwidth} of $G$ is the minimum width over all its path decompositions.

In order to prove \cref{th:pathwidth}, we begin with the following lemma.

\begin{lemma}\label{le:pathwidth}
	Let $G=(V,E)$ be a graph. If $G$ admits a \story $\mathcal{S}=\langle \tau, \{D_i\}_{i\in [n]} \rangle$ of width $\omega=w(\mathcal{S})$, then it admits a path decomposition of width at most $\omega$.
\end{lemma}
\begin{proof}
	It is easily seen that $\omega$ equals the vertex separation number of the total order $\tau$. Namely, suppose for a contradiction that there exists an index $i \in [n]$ such that there are at least $\omega+1$ vertices $v$ with $\tau(v) \le i$ having a neighbor $u$ such that $\tau(u) > i$. Let $u_1$ be the one of such neighbors that appears first. All the at least $\omega+1$ vertices are still visible at step $\tau(u_1)$ and therefore the size of the frame at step $\tau(u_1)$ is at least $\omega+2$, thus contradicting the fact that $\mathcal{S}$ has with $\omega$. On the other hand, a total order $\tau$ with vertex separation number $\omega$ implies the existence of a path decomposition of width at most $\omega$~\cite{DBLP:journals/ipl/Kinnersley92}. 
\end{proof}

\thpathwidth*
\begin{proof}
	If $G$ does not admit a \story, for any total order of its vertices, then $\pw(G) < \fw(G) = +\infty$.
	Otherwise, $\pw(G) \le \fw(G)$ by \cref{le:pathwidth}.
	
	If $G$ is planar, we show that $\fw(G) \le \pw(G)$. 
	Namely, let $\{X_i\}_{i \in [h]}$ be a path decomposition of $G$. 
	Without loss of generality, we can assume that this path decomposition is \emph{nice} (or normalized), i.e., $X_1$ contains exactly one vertex, and, for each bag $X_i$ with $1<i \le h$, there is a vertex $v$ such that either $X_i=X_{i-1} \cup \{v\}$ or $X_i=X_{i-1} \setminus \{v\}$. Then we have that at each step at most one vertex is introduced in the decomposition, and hence we can associate each vertex $v$ with an  index $l_v$ such that $l_v=1$ if $v$ is the first introduced vertex, and $l_v = l_u+1$, where $l_u$ is the last vertex introduced before $v$, otherwise. We define a total order $\tau$ of $V$ by setting $\tau(v)=l_v$. Let $\Gamma$ be a planar drawing of $G$. Consider the set of $n$ drawings such that $D_i$ corresponds to the subdrawing of $\Gamma$ induced by the vertices that are visible at step $i$, for each $i \in [n]$. One immediately verifies that $\langle \tau, \{D_i\}_{i\in [n]} \rangle$ is a \story of $G$ of width~$\pw(G)$.
\end{proof}

In order to prove \cref{le:bipartite-3}, we first prove a few auxiliary lemmas.

\begin{lemma}\label{le:bipartite-1}
	Let $K_{a,b} = (A \cup B, E)$ be a complete bipartite graph with $a = |A|$ and  $b=|B|$. 
	Let $\mathcal{S}=\langle \tau, \{D_i\}_{i\in [a+b]} \rangle$ be a \story of $K_{a,b}$. 
	There exists $i \in [a+b]$ such that all vertices of $A$ or $B$ are visible. 
\end{lemma}
\begin{proof}
	Let $i$ be such that $D_i$ contains the largest number $t$ of vertices of $A$ over all frames of $\mathcal{S}$. 
	If $t=a$, we are done. 
	If $t<a$, there exist two vertices $u,v$ of $A$ such that  $j_u < i_v$. Note that all vertices in $B$ are adjacent to $u$, and hence they all appear at some step smaller than or equal to $j_u$. On the other hand,  since all vertices in $B$ are adjacent to $v$ as well, they cannot disappear before $i_v+1$. It follows that all vertices of $B$ are visible at step $j_u$.
\end{proof}


\begin{lemma}\label{le:bipartite-2}
	Let $K_{a,b} = (A \cup B, E)$ be a complete bipartite graph with $a = |A|$ and  $b=|B|$. 
	Let $\mathcal{S}=\langle \tau, \{D_i\}_{i\in [a+b]} \rangle$ be a \story of $K_{a,b}$. 
	If a step $i \in [a+b]$ contains at least three vertices of $A$ (resp. $B$), then it contains at most two vertices of $B$ (resp. $A$).  
\end{lemma}
\begin{proof}
	The statement immediately follows from the fact any frame of $\mathcal{S}$ is planar and hence cannot contain $K_{3,3}$ as a subgraph.  
\end{proof}

We are now ready to prove \cref{le:bipartite-3}.

\lebipartitethree*
\begin{proof}
	Consider the interval $I=[s,t] \subseteq [a+b]$ of maximal length such that the vertices of $A$ or $B$, say $A$, are all visible. 
	By \cref{le:bipartite-1}, $I$ is not empty.  Let $j$ be one of the steps that contain the largest number $h$ of vertices of $B$. Observe that, since $I$ is maximal, one vertex of $A$ appears at $s$. Therefore, any vertex of $B$ visible at a step smaller than $s$ is visible also at step $s$. Similarly, by the maximality of $I$, one vertex of $A$ disappears at step $t+1$, therefore any vertex of $B$ visible at a step greater than $t$ is visible also at step $t$. Consequently, we can assume that $j \in I$, and we can conclude $h \le 2$ by \cref{le:bipartite-2}.
\end{proof}

\clearpage

\section{Missing Proofs of \cref{se:storyp}}

\subsection{Missing proof of \cref{sse:hardness}}\label{ap:hardness}

\begin{figure}[t]
	\centering
	\includegraphics[scale=1,page=5]{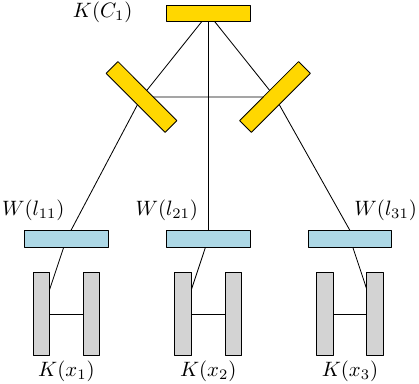}
	\caption{Illustration for \cref{th:hardness}. Schematization of $(x_1 \vee x_2 \vee \neg x_3)$.\label{fi:hard-2}}
\end{figure}

\lecorrectnessone*
\begin{proof}
	Let $\mathcal{S}=\langle \tau, \{D_i\}_{i \in [n]} \rangle$ be a \story of $G$. For each variable gadget $K(x_i)$ we assign the  value \emph{true} to $x_i$ if the v-side $A_i$  is flexible in $\mathcal{S}$. Consider any literal $l_{ij}$ and the wire gadget $W_{ij}$. If $l_{ij}$ is a positive (negative) literal, then $A_i$  ($B_i$) and $W_{ij}$ form a $K_{3,3}$, hence by \cref{le:bipartite-3} the w-side $W_{ij}$ is fixed (flexible). Analogously, if we consider the clause gadget $K(C_j)$, the c-side connected with $W_{ij}$ is flexible (fixed). Symmetrically, we assign the value \emph{false} to $x_i$ if the v-side $B_i$ is instead flexible in $\mathcal{S}$, and for any positive (negative) literal $l_{ij}$, the w-side $W_{ij}$ is flexible (fixed), while the corresponding c-side of $K(C_j)$ is fixed (flexible).  In other words, the value of $x_i$ propagates consistently throughout all its literals. It remains to prove that, for any clause $C_j$ of $\varphi$, precisely one literal is true for the constructed truth value assignment of $\{x_i\}_{i \in [N]}$. Namely, we claim that exactly one c-side of $K(C_j)$ is flexible, while the other two are fixed. 
	
	We first argue that not all c-sides can be fixed in $\mathcal{S}$. Assume, for a contradiction, that they are. Let $D_h$ be the frame of $\mathcal{S}$ in which the last simple vertex $v$ of $K(C_j)$ appears. Observe that all other simple vertices are also visible at step $h$. Namely, let $u \neq v$ be a simple vertex of $K(C_j)$. Either $u$ is adjacent to $v$ or in the same c-side as $v$. Hence, since $i_u < h$ by assumption and since $u$ cannot disappear until $h+1$, it follows that $u$ is visible at step $h$. It follows that $D_h$ contains a drawing of $K_{2,2,2}$, which is a maximal planar graph and hence has a unique planar embedding up to the choice of the outer face. In particular, vertex $v$ and the other vertex on the same c-side as $v$, which we call $v'$, do not belong to the same face. Consider now the special vertex $s$ adjacent to both $v$ and $v'$. Vertex $s$ cannot be visible at step $h$, else $D_h$ would not be planar, as it would contain $K_{2,2,2}$ plus a subdivided edge between $v$ and $v'$. On the other hand, it cannot have disappeared already, since $s$ and $v$ are adjacent. Hence, $s$ has not appeared yet. Similarly, no other special vertex is visible at step $h$, and since the special vertices are pairwise adjacent, it follows that none of them has appeared yet. Consequently, no vertex of $K_{2,2,2}$ can disappear until the first special vertex has appeared, which implies the existence of a frame that contains a planar drawing of $K_{2,2,2}$ plus a subdivided edge between two vertices in the same c-side, which contradicts the fact that the considered frame is planar. 
	
	We now rule out the case in which more than one c-side is flexible, which corresponds to a clause with more than one true literal. If a c-side is flexible, there can be no frame containing both its simple vertices. This follows from the fact that the special vertex is adjacent to both simple vertices, and hence there would be a frame containing all three vertices together, which is not possible since the c-side is flexible. On the other hand, when the last simple vertex appears, its four neighbors in $K_{2,2,2}$ are all visible, and hence at least two c-sides must be fixed. 
	
	Altogether, we have proved that at least two c-sides are fixed and that at least one c-side is flexible. Therefore, in each clause gadget, exactly one c-side if flexible, which corresponds to having exactly one true literal, as desired. 
\end{proof}
\lecorrectnesstwo*
\begin{proof}

	We avoid repeating the initial part of the proof, which is reported in the main body of the paper. Observe that, after having drawn the fixed v-sides and w-sides, each flexible v-side can be drawn independently, by letting appear and disappear its three vertices one by one. When a vertex appears, its edges towards the fixed w-sides of its wire gadgets and the fixed v-side of its variable gadget form a star; see \cref{fi:draw-flex-variable}. To conclude the proof, we first draw the clause gadgets alone, and then we show how to integrate in the \story the flexible w-sides and their connections, as well as the connections of the fixed w-sides with the corresponding c-sides.

	We now draw the clauses one by one. We begin by showing how to draw a single clause gadget $K(C_j)$, ignoring the connections with the linked wire gadgets. Refer to \cref{fi:draw-clause} for an illustration.  We first let appear the three special vertices of $K(C_j)$ (\cref{fi:draw-clause-a}). Now we let appear the two simple vertices of a false literal, which we denote by $l_1$ for convenience (\cref{fi:draw-clause-b}). Note that the whole c-side representing $l_1$ is visible. Right after, the special vertex connected to the two simple vertices can disappear. Next, we let appear the two simple vertices of the remaining false literal, called $l_2$, hence again this c-side is entirely visible (\cref{fi:draw-clause-c}). Again this can be done without crossings, and the corresponding special vertex can disappear right after. Now we have a $4$-cycle (with a light blue background in \cref{fi:draw-clause-c}) formed by the edges connecting the simple vertices drawn so far, and the special vertex of the missing literal lying in one of the two faces made by the cycle. We let one of the two simple vertices of the true literal appear (\cref{fi:draw-clause-d}), we draw its edges planarly, and afterwards we let it disappear and replace it with the last simple vertex (\cref{fi:draw-clause-e}). Once the last simple vertex has appeared, all vertices of the clause can disappear. Hence, the last c-side hasn't appeared together but instead acted as a flexible c-side.

	\begin{figure}[p]
		\centering
		\subfigure[]{\includegraphics[scale=0.7,page=15]{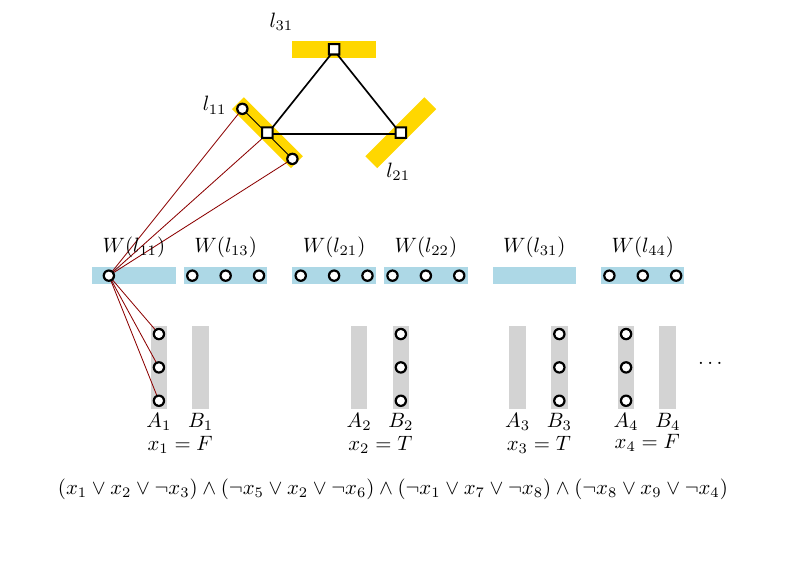}\label{fi:draw-clause-wire-var-a}}
		\subfigure[]{\includegraphics[scale=0.7,page=14]{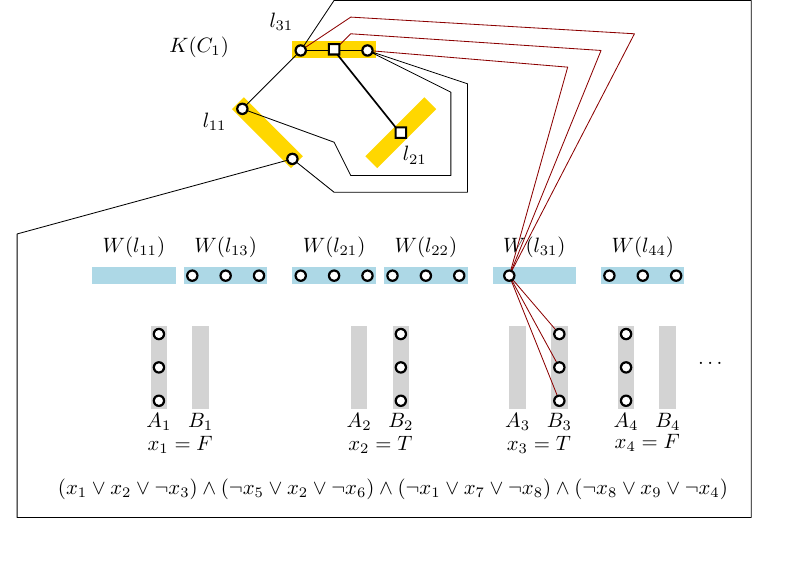}\label{fi:draw-clause-wire-var-b}}
		\caption{Proof of \cref{le:correctness-2}: drawing the false literals.}
	\end{figure}

	\begin{figure}
		\centering
		\includegraphics[scale=0.7,page=16]{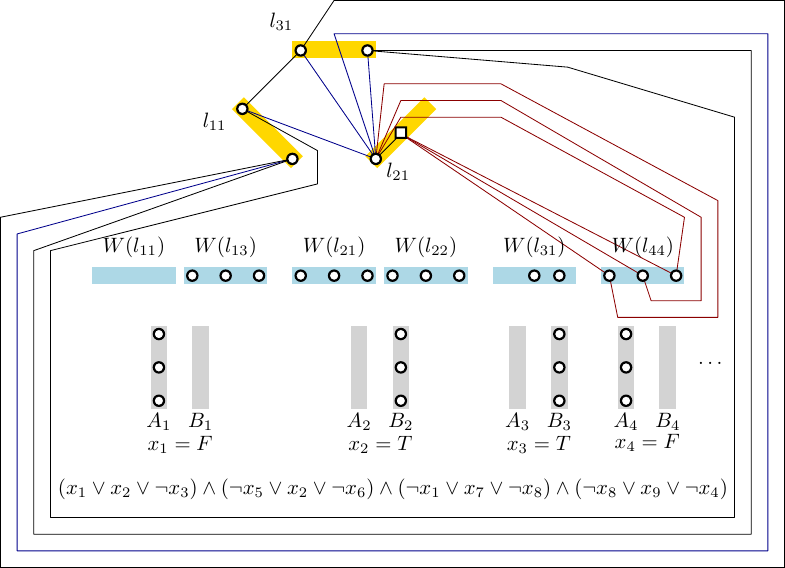}
		\caption{Proof of \cref{le:correctness-2}: drawing the true literal.\label{fi:draw-clause-wire-var-c}}
	\end{figure}
	
	Consider a wire gadget $W(l_{ij})$. Suppose first that $l_{ij}$ is false. Then we have seen that its c-side in $K(C_j)$, say $U$, is fixed and hence $W(l_{ij})$ must be flexible. Note that both if $l_{ij}=l_1$  (\cref{fi:draw-clause-wire-var-a}) or $l_{ij}=l_2$ (\cref{fi:draw-clause-wire-var-b}), in the frame in which the three vertices of $U$ are visible there is a face whose boundary contains the three vertices of the corresponding c-side. We can assume that this face encloses the drawings of all the w-sides and v-sides that are currently visible. Hence the vertices of $W(l_{ij})$ can be drawn one by one in this face. If $l_{ij}$ is true, then its c-side $U$ is flexible and hence $W(l_{ij})$ must be fixed. Again in the two frames in which the vertices of $U$ appear (see \cref{fi:draw-clause-d,fi:draw-clause-e}), the drawing of the clause has a face whose boundary contains the vertex of $U$ visible in that frame. Again we can assume that this face encloses the drawings of all the w-sides and v-sides that are currently visible, and therefore the vertices of $U$ can be connected planarly both to its neighbors in $K(C_j)$ and to its neighbors in $W(l_{ij})$ (see \cref{fi:draw-clause-wire-var-c}). By repeating this procedure for each clause we complete the \story.
\end{proof}

\thhardness*
\begin{proof}
	Constructing the graph $G$ from the formula $\varphi$ clearly takes polynomial time, and the correctness of the reduction follows from \cref{le:correctness-1,le:correctness-2}. This proves that \storyp is \NP-hard.
	
	For the second part of the statement, observe that the \emph{\ethlong} (\ethshort), combined with the Sparsification Lemma, states that \sat formulas on $N$ variables and $M$ clauses cannot be solved in $2^{o(N+M)}$~time, see~\cite{DBLP:journals/jcss/ImpagliazzoP01,DBLP:journals/jcss/ImpagliazzoPZ01}. On the other hand, there exists a polynomial-time reduction from \sat to \onesat that transforms \sat instances with $N+M$ variables and clauses into equivalent \onesat instances with $O(N+M)$ variables and clauses~\cite{DBLP:conf/stoc/Schaefer78}.  Finally, starting from an instance of \onesat, our reduction constructs a graph with $n$ vertices and $m$ edges such that $m=O(n)$ and $n+m=O(N+M)$. Therefore, an algorithm solving \storyp in $2^{o(n)}$ time would contradict ETH.
\end{proof}

\thupper*
\begin{proof}
	The key observation to claim that the problem belongs to \NP lies in the fact that, since we consider topological drawings in which edges are Jordan arcs, once a total order of the vertices is fixed, the existence of a \story is equivalent to the existence of a sequence of planar embeddings that are ``consistent''. Namely, for any two such planar embeddings, if they share a common subgraph, then their restrictions to this subgraph must coincide. With this observation, we can guess a candidate solution and verify, in polynomial time, if it is valid. However, we intertwine the generation process of the solutions and their validity check in order to obtain a more efficient time complexity. 
	
	More formally, we proceed as follows. We first guess a linear order of the vertices of $G$. Note that there are $n! \in 2^{O(n \log n)}$ such orders. For each $i \in [n]$, we identify the set of vertices $V_i$ visible at step $i$ and the corresponding set of edges $E_i$. By planarity, it must be $|E_i| \le 3|V_i|-6$, else we can immediately reject the instance. As already observed, we do not guess a drawing $D_i$, but instead a planar embedding $\mathcal{E}_i$ for graph $G_i=(V_i,E_i)$. 
	
	If $G_i$ is connected, the planar embedding $\mathcal{E}_i$ can be described by the circular order of the edges around each vertex, called \emph{rotation system}, and by the choice of the outer face. Indeed, with this information, one can reconstruct the boundary of each face of $G_i$ (in linear time). If $G$ is not connected, we also need the pairwise relative positions of distinct connected components (if more than one). Let $C_1,\dots,C_q$ be the connected components of $G_i$ and let $\mathcal{E}_i(C_j)$ be the restriction of $\mathcal{E}_i$ to $C_j$ for each $j \in [q]$. The \emph{position system} is a rooted tree whose vertices are the faces of the embeddings  $\mathcal{E}_i(C_j)$ for each $j \in [q]$, and a node $\nu$ has a parent $\mu$ in this tree if and only if the face corresponding to $\nu$ contains the face corresponding to $\mu$ in its interior and $\mu$ is the outer face of its connected component. 
	
	\noindent The number of possible rotation systems of $G_i$, denoted by $n_{\textrm{rot}}$, is upper bounded by the number of possible permutations of edges around each vertex.  Thus we have 
	
	$$n_{\textrm{rot}} \le \prod_{v \in V_i} {\textrm{deg}(v)!} < \prod_{v \in V_i} {\textrm{deg}(v)^{\textrm{deg}(v)}} \in n^{O\left(\sum_{v \in V_i}{\textrm{deg}(v)}\right)} \subseteq n^{O(n)}.$$

	\noindent Each rotation system of $G_i$ fixes the boundary of each face of each connected component of $G_i$. Hence we can verify (in linear time) whether the rotation system yields a planar combinatorial (up to the choice of the outer face) embedding of each connected component of $G_i$. If this is the case, there are $O(n)$ possible outer faces for each connected component, we have $O(n^2)$ possible choices over all connected components. Once both the rotation system and the outer face of each connected component have been fixed, the number $n_{\textrm{pos}}$ of possible position systems is $n_{\textrm{pos}} \le n^{O(n)}$, since it suffices to guess the parent of each node of the tree. Therefore, for each  step $i>1$, there are $O(n^2) \cdot n_{\textrm{rot}} \cdot n_{\textrm{pos}} \subseteq n^{O(n)} \subseteq 2^{O(n \log n)}$  distinct embeddings that we can guess. 
	
	Once we have generated all planar embeddings for each $i \in [n]$, we aim at finding a sequence (if any) of planar embeddings, one for each step $i$, in which common subgraphs have the same embedding throughout the sequence. Namely, for each $i>1$ and for each planar embedding $\mathcal{E}_i$, we verify whether there is at least one planar embedding $\mathcal{E}_{i-1}$ such that the restrictions of $\mathcal{E}_i$ and $\mathcal{E}_{i-1}$  to the common graph $G_{i-1} \cap G_i$ are the same. If this is the case, we keep $\mathcal{E}_i$, else we discard it. The algorithm halts if the set of planar embeddings becomes empty, else it outputs at least one valid sequence. Overall, the procedure takes $2^{O(n \log n)} \cdot 2^{O(n \log n)}= 2^{O(n \log n)}$ time.
\end{proof}

\subsection{Missing proofs of \cref{sse:vc}}\label{apx:vcn}

\lekernel*
\begin{proof}
	One direction is trivial, since it is easily seen that admitting a \story is a hereditary property. 
	Suppose now that $G'$ admits a \story $\mathcal{S'}=\langle \tau', \{D'\}_{i\in [n']}\rangle$, where $n'=n-1$. 
	We distinguish three cases based on the reduction rule applied to $G$. In each case, we denote by $x$ the vertex removed from $G$ to obtain $G'$.
	
	\textbf{Case A (R.1).} See \cref{fi:fpt-1} for an illustration. Let $v$ be the neighbor of $x$ in $G$, whose lifespan according to $\tau'$ is $[i_v,j_v]$. We compute $\tau$ from $\tau'$ by inserting $x$ right after $v$, consequently the lifespan of $x$ in $\tau$ is $[i_v+1,i_v+1]$. Similarly, compute $\{D_i\}_{i\in [n]}$ from $\{D'\}_{i\in [n']}$ as follows. For each $i \le i_v$, we set $D_i = D'_i$. For $i=i_v+1$, we draw $x$ in $D'_{i_v}$ sufficiently close to $v$ such that $xv$ can be drawn as a straight-line segment that does not intersect any other edge. We then set $D_i$ to be equal to the resulting drawing. Finally, for each $i>i_v+1$, we set $D_i = D'_{i-1}$.
	
	\begin{figure}
		\centering
		\includegraphics[width=\textwidth,page=2]{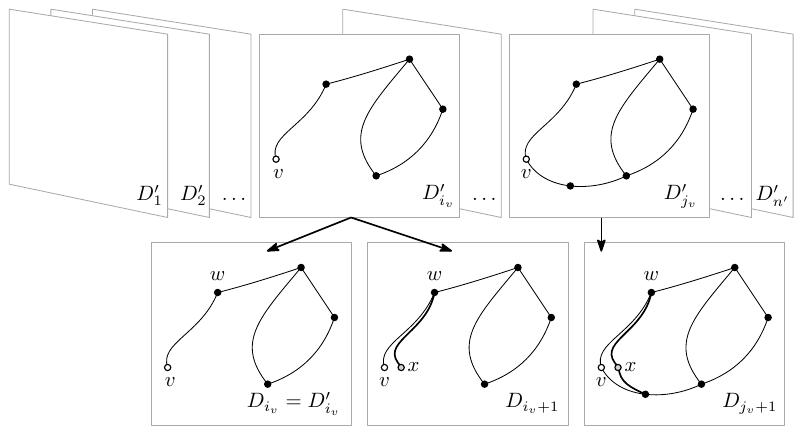}
		\caption{Illustration for \textbf{Case B} of the proof of \cref{le:kernel}.\label{fi:fpt-2}}
	\end{figure}
	
	\textbf{Case B (R.2).} See \cref{fi:fpt-2} for an illustration. By assumption, $G'$ contains at least one vertex $v \neq x$ of type $U$, whose lifespan according to $\tau'$ is $[i_v,j_v]$. We compute $\tau$ from $\tau'$ by inserting $x$ right after $v$; consequently, the lifespan of $v$ in $\tau$ is $[i_v+1,j_v]$. For each $i \le i_v$, we set $D_i = D'_i$. For $i=i_v+1$, we extend $D'_i$ by drawing $x$ sufficiently close to $v$ and by drawing, for each neighbor $w$ of $x$, the edge $xw$ such that it follows the curve representing the edge $vw$. Since $vw$ does not cross any other edge and $v$ has degree two, the same holds for $xw$.  We then set $D_i$ to be equal to the resulting drawing. Similarly, for each $i \le [i_v+2,j_v+1]$, we extend (if needed) the frame $D'_{i-1}$ by drawing any edge $xw$ such that it follows the corresponding curve $vw$. Finally, for each $i>j_v+1$, we set $D_i = D'_{i-1}$.
	
	\textbf{Case C (R.3).} Observe that, by assumption, the graph induced by the vertices of $U \cup V_U$ contains a complete bipartite graph $K_{a,b}$ with $a \ge 3$ and $b \ge 3$, whose partite sets are $U$ and $V_U$. We distinguish two subcases based on whether $U$ is the fixed set or the flexible set of $K_{a,b}$. 
	
	\begin{figure}
		\centering
		\includegraphics[width=\textwidth,page=3]{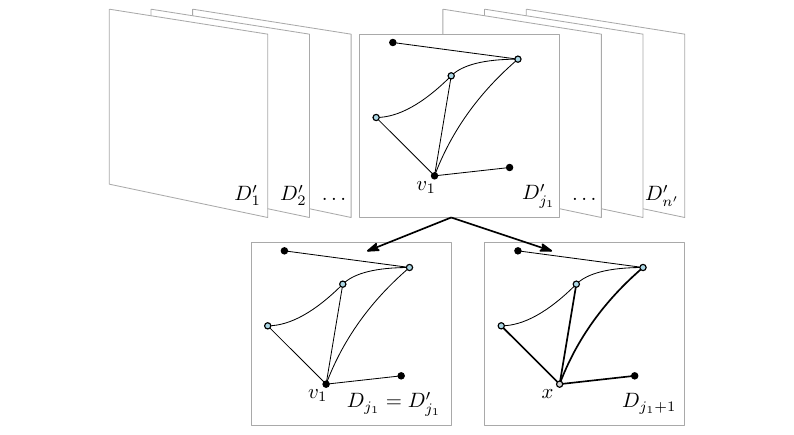}
		\caption{Illustration for \textbf{Case C} of the proof of \cref{le:kernel}, when $U$ is the fixed set.\label{fi:fpt-3a}}
	\end{figure}
	
	Suppose first that $U$ is the fixed set; see \cref{fi:fpt-3a} for an illustration.  Let $I \in [n']$ be the interval in which all vertices of $U$ are visible. By assumption, $G'$ contains at least three vertices $v_j \neq x$ of $V_U$, with $j=1,2,3$. Observe that the lifespan of each vertex $v_j$ intersects $I$, therefore there are at least two vertices, say $v_1$ and $v_2$, whose lifespans do not intersect, otherwise there would be a frame containing $K_{3,3}$. Let $[i_1,j_1]$ and $[i_2,j_2]$ be the lifespan of $v_1$ and $v_2$, respectively, and suppose (without loss of generality) that $j_1 < i_2$. Observe that both $j_1$ and $i_2$ are in $I$.  We compute $\tau$ from $\tau'$ by inserting $x$ right after $v_1$ disappears, such that its lifespan in $\tau$ is $[j_1+1,j_1+1]$. 
	For each $i \le j_1$, we set $D_i = D'_i$. For $i=j_1+1$,  we take $D'_{j_1}$ and replace the drawing of $v_1$ with the drawing of $x$. Namely, we place $x$ on the same point of $v_1$ and we draw each curve $xw$ by following the curve $v_1w$. Since $v_1w$ does not cross any other edge, the same holds for $xw$. The resulting drawing is $D_i$. Finally, for each $i>j_1+1$, we set $D_i = D'_{i-1}$.
	
	\begin{figure}
		\centering
		\includegraphics[width=\textwidth,page=4]{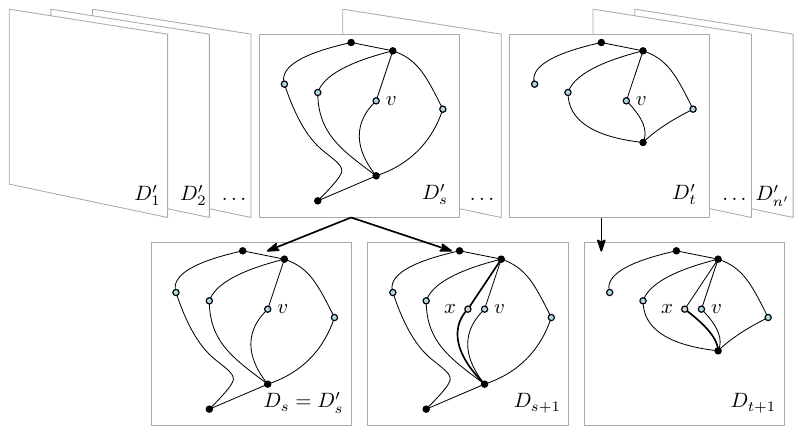}
		\caption{Illustration for \textbf{Case C} of the proof of \cref{le:kernel}, when $V_U$ is the fixed set.\label{fi:fpt-3b}}
	\end{figure} 
	
	Suppose now that $V_U$ is the fixed set; see \cref{fi:fpt-3b} for an illustration. Let $[s,t] \in [n']$ be the maximal interval in which all vertices of $V_U$ in $G'$ are visible.  Let $v \neq x$ be the vertex of $V_U$ such that $\tau'(v)=s$. We compute $\tau$ from $\tau'$ by inserting $x$ right next $v$ such that its lifespan in $\tau$ is $[s+1,t+1]$. For each $i \le s$, we set $D_i = D'_i$. For $i=s+1$,  we extend $D'_i$ by drawing $x$ sufficiently close to $v$ and by drawing, for each neighbor $w$ of $x$, the edge $xw$ such that it follows the curve representing the edge $vw$. Note that  in each frame $D'_i$ with $i \in [s,t]$ there are at most two vertices of $U$, and that there are no edges between vertices of $V_U$. Hence, similarly as in \textbf{Case B}, vertex $v$ has degree two and the curve $xw$ does not cross any other edge.  We then set $D_i$ to be equal to the resulting drawing. Similarly, for each $i \le [s+2,t+1]$, we extend (if needed) the frame $D'_{i-1}$ by drawing any edge $xw$ such that it follows the corresponding curve $vw$. Finally, for each $i>t+1$, we set $D_i = D'_{i-1}$.
\end{proof}

\thfptvc*
\begin{proof}
	By~\cite{DBLP:journals/tcs/ChenKX10}, we can determine the vertex cover number $\vc=\vc(G)$ of $G$ and compute a vertex cover $C$ of size $\vc$ in time $O(2^\vc+\vc\cdot n)$.  We construct a kernel $G^*$ from $G$ of size $O(2^{\vc})$ as follows. We first classify each vertex of $G$ based on its type. We then apply reduction rules \textbf{R1}, \textbf{R2}, and \textbf{R3} exhaustively.  Thus, constructing $G^*$ can be done in $O(2^\vc + \vc \cdot n)$ time, since $O(2^\vc)$ is the number of types and $\vc \cdot n$ is the maximum number of edges of $G$. Also, $G^*$ contains $n^* \le 3 \cdot 2^\vc < 2^{\vc+2}$ vertices.

	From \cref{le:kernel} we conclude that $G$ admits a \story if and only if $G^*$ does. To establish whether $G^*$ admits a \story we proceed as follows: (1) We guess a total order $\tau^*$ of $G^*$; (2) For $i=1$, we guess all planar embeddings of the graph induced by the vertices visible at step $i$; (3) For each $i>1$, we consider the embeddings computed at the previous step $i-1$, we remove from them the vertices (if any) that disappear at step $i$, we remove possible duplicates, and we try to exhaustively extend each of the resulting planar embeddings by inserting the vertex that appears at step $i$. The algorithm halts if the set of planar embeddings becomes empty.  It is readily seen that $G^*$ admits a \story if and only if the algorithm terminates at step $n^*$ with at least one planar embedding. Concerning the time complexity, step (1) takes $O(2^{\vc+2}!)$ time. Since $G^*$ contains $O(2^{O(\vc)})$ vertices and edges, the number of possible planar embeddings are $O((2^{O(\vc)})^{(2^{O(\vc)})})= O((2^{2^{O(\log \vc)}\cdot 2^{O(\vc)}})= O(2^{2^{O(\vc)}})$. Hence step (2) takes $O(2^{2^{O(\vc)}})$ time and step (3) takes $O(2^{2^{O(\vc)}}) \cdot O(2^{2^{O(\vc)}}) = O(2^{2^{O(\vc)}})$ time. Starting from a \story of $G^*$, we can reinsert the missing $O(n)$ vertices each in $O(n)$ time, as detailed in \cref{le:kernel}.
\end{proof}

\subsection{Missing proofs of \cref{sse:fes}}\label{ap:fes}

In order to prove \cref{le:kernel-fes-size}, we first show that \textbf{R.A} and \textbf{R.B} are safe rules.

\begin{lemma}\label{le:kernel-fes}
	Let $G'$ be the graph obtained from $G$ by applying one of the reduction rules \textbf{R.A} and \textbf{R.B}. 
	Then $G$ admits a \story if and only if $G'$ does. 
\end{lemma}
\begin{proof}
	One direction is again trivial, since admitting a \story is a hereditary property. 
	Suppose now that $G'$ admits a \story $\mathcal{S'}=\langle \tau', \{D'\}_{i\in [n']}\rangle$, where $n'<n$. 
	We distinguish two cases based on the reduction rule applied to $G$. 
	
	\textbf{Case 1 (R.A). } The argument is analogous as for rule \textbf{R.1}. Let $u$ be the removed vertex and let $v$ be its neighbor in $G$, whose lifespan according to $\tau'$ is $[i_v,j_v]$. We compute $\tau$ from $\tau'$ by inserting $u$ right after $v$, consequently the lifespan of $v$ in $\tau$ is $[i_v+1,i_v+1]$. Similarly, compute $\{D_i\}_{i\in [n]}$ from $\{D'\}_{i\in [n']}$ as follows. For each $i \le i_v$, we set $D_i = D'_i$. For $i=i_v+1$, we draw $u$ in $D'_{i_v}$ sufficiently close to $v$ such that $uv$ can be drawn as a straight-line segment that does not intersect any other edge. We then set $D_i$ to be equal to the resulting drawing. Finally, for each $i>i_v+1$, we set $D_i = D'_{i-1}$.
	
	\textbf{Case 2 (R.B). } Let $\langle s,u_1,\dots,u_k,t \rangle$ be the $k$-chain whose inner vertices, $u_1,\dots,u_k$, have been removed. Let $[i_s,j_s]$ and $[i_t,j_t]$ be the lifespans of $s$ and $t$ according to $\tau'$, respectively. Without loss of generality, we assume $i_s < i_t$. We compute $\tau$ from $\tau'$ as follows. We insert $u_1$ right after $s$ and $u_k$ right after $t$. Now the order has length $n'+2$.  For each $i\in\{2,\dots,k-1\}$ (recall that $k \ge 3$), we insert $u_i$ at step $[n'+i+1]$. Consequently, the lifespans of $u_1$ and $u_k$ in $\tau$ are $[i_s+1,n'+3]$ and $[i_t+1,n'+k]$, respectively,  the lifespan of $u_i$ is $[ n'+i+1,n'+i+2]$, for each $i\in\{2,\dots,k-2\}$, and the lifespan of $u_{k-1}$ is $[n'+k,n'+k]$. For each $i \le i_s$, we set $D_i = D'_i$. For $i=i_s+1$, we draw $u_1$ in $D'_{i_s}$ sufficiently close to $s$ such that $su_1$ can be drawn as a straight-line segment that does not intersect any other edge. Analogously, for $i=i_t+1$, we draw $u_{k}$ in $D'_{i_t}$ sufficiently close to $t$ such that $u_kt$ can be drawn as a straight-line segment that does not intersect any other edge. Recall that $k \ge 3$ and hence edge $u_1u_k$ does not exist, which implies that all frames up to $D_{n'+2}$ are planar. For each $i \in \{n'+3,n'+k\}$ we  have that vertex $u_j$, with $j=i-n'-1$ appears, and the only visible vertices are $u_{j-1}$, $u_{j}$, and $u_k$. In particular, the only edges to be drawn are $u_{j-1}u_j$ and $u_ju_k$ if $j=k-1$ (i.e., if $i=n'+k$), hence they can be realized as crossing-free straight-line segments for any position of $u_j$ (avoiding the three vertices being collinear). 
\end{proof}

\lekernelfessize*
\begin{proof}
	Let $G$ be a graph with $n$ vertices and let $F$ be a feedback edge set of $G$ of size $\fes=\fes(G)$.
	Let $G^*$ be the graph obtained after applying exhaustively the reduction rules \textbf{R.A} and \textbf{R.B}.  Graph $G^*$ is a kernel of $G$ by \cref{le:kernel-fes}.
	Let $F^* \subseteq F$ be the edges of $F$ that belong to $G^*$. Observe that removing the edges of $F^*$ from $G^*$ yields a forest $H^*$ with at most $2\fes$ leaves in total. Indeed, any leaf of $H^*$ is an endpoint of some edge in $F^*$, as otherwise it would not be part of $G^*$ by \textbf{R.A}. Moreover, by \textbf{R.B}, $G^*$ (and hence $H^*$) contains a $k$-chain only if $k<3$. Consequently, $H^*$ contains at most $16\fes -4$ vertices, namely at most $2\fes$ leaves, $(2\fes-1)$ vertices of degree larger than two, and $3(4\fes-1)$ vertices of degree two. Therefore, $G^*$ has $O(\fes)$ vertices and edges.
\end{proof}

\subsection{Missing Proofs for \cref{sse:threetrees}}\label{ap:3tree}
\ththreetrees*
\begin{proof}

	\begin{figure}[htbp]
		\centering
		\includegraphics[width=0.45\linewidth,page=4]{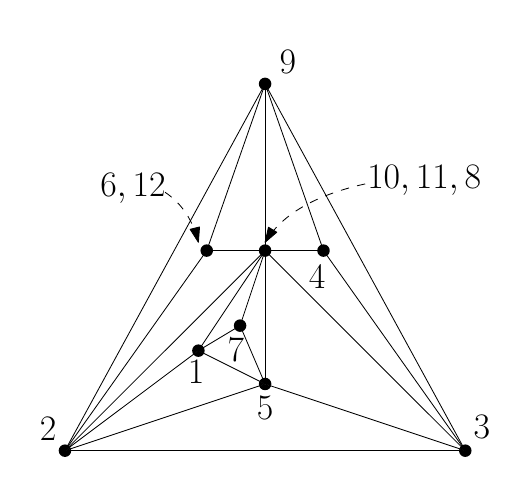}
		\caption{The graph $G^*$ uf the $3$-tree $G$ of \Cref{fi:3-trees}; each vertex $v$ of $G^*$ is labeled with the vertices of $G$ whose representative vertex is $v$. \label{fi:3-trees-d}}
	\end{figure}
	
	We avoid repeating the part of the proof described in the main body of the paper. We only prove that the described \story can be computed in linear time. If $G$ is a partial $3$-tree it can be augmented to a (non-partial) $3$-tree in $O(n)$ time~\cite{MATOUSEK19911}. The reverse of the order in which the vertices are added to construct $G$ is called a \emph{perfect elimination ordering} and can be computed in $O(n)$ time~\cite{doi:10.1137/0205021}. Starting from this order we can compute $\T$ in $O(n)$ time. The order $\tau$ can be computed in $O(n)$ time by performing a preorder visit of $\T$. 
	During this visit we also compute some information that are useful to compute efficiently all the drawings $\{D_i\}_{i\in [n]}$. First, we compute an embedded planar $3$-tree $G^*$ that is a supergraph of each $G_i$ (refer to \Cref{fi:3-trees-d}). The idea to construct $G^*$ is to keep only one representative, for each set of vertices whose parent cliques have the same set of representatives. We construct $G^*$ starting with the initial $3$-cycle (which is associated with the root); each vertex of this $3$-cycle is its own representative. When we visit a bag $\mu$, we consider the three representatives of the vertices of the parent clique $C$ of $v_{\mu}$; these three vertices can be obtained in $O(1)$ time by storing for each vertex its representative. If a vertex $v$ connected to the representative vertices of $C$ already exists, we do not modify $G^*$ and store $v$ as the representative vertex of $v_{\mu}$; if no vertex has already been connected to the representative vertices of $C$, we add a new vertex $v$ to $G^*$ connecting it to the representative vertices of $C$ and, also in this case, we store $v$ as the representative vertex of $v_{\mu}$. Since both cases can be handled in $O(1)$ time, the overall construction of $G^*$ can be executed in $O(n)$ time. Also, $G^*$ has $O(n)$ vertices and edges.
	Second, we store for each $i$ (with $i=1,2,\dots,n$) the list $\ell_i$ of vertices that have to be removed when constructing $D_i$ from $D_{i-1}$. These vertices are all the vertices associated with the bags of $\T$ that are visited before $\mu_i$ and after its parent $\mu_j$ (this set of vertices may be empty) and they can be stored, in the reverse order of addition, while backtracking from $\mu_{i-1}$ during the preorder visit.
	
	Once we have $G^*$ we can use any existing $O(n)$-time algorithm to compute a planar straight-line grid drawing $D^*$ of $G^*$ (for example the shift algorithm~\cite{CHROBAK1995241}). By the construction of $G^*$, each drawing $\{D_i\}_{i\in [n]}$ is a sub-drawing of $D^*$ and therefore we can compute each $D_i$ by removing from $D^*$ the vertices and edges that do not belong to $G_i$. In particular at each step $i$ we can use the list $\ell_i$ to find the vertices to be removed from $D_{i-1}$; these vertices are removed in reverse order of addition and when a vertex is removed we also remove the three edges that connect it to its parent clique. Since each vertex and each edge is added once and removed once, the set $\{D_i\}_{i\in [n]}$ can be computed in $O(n)$ time.
\end{proof}

\clearpage

\section{Missing Proofs for \cref{se:storypfo}}

\thfohardness*
\begin{proof}
	Suppose first that $G_1, G_2, G_3$ admit a SEFE $\Gamma_1,\Gamma_2,\Gamma_3$. Let $\Gamma$ be the drawing obtained as follows: (1) let $\Gamma'$ be the union of $\Gamma_1,\Gamma_2,\Gamma_3$; (ii) let $\Gamma''$ be the drawing obtained from $\Gamma'$ by subdividing all exclusive edges in $E_i$ for each $i=1,2,3$ (the subdivision vertex can be placed at any point along the curve representing the edge); (iii) finally, let $\Gamma$ be the drawing obtained from $\Gamma''$ by drawing for each vertex $v$ its spectator $s_v$ sufficiently close such that the edge $vs_v$ can be realized as a straight-line segment that does not intersect any other edge.  Then a \story $\langle \tau, \{D_i\}_{i\in [n]} \rangle$ can be obtained by setting each drawing $D_i$ equal to the subdrawing of $\Gamma$ induced by the vertices visible at step $i$. By construction, the position of a vertex  is the same over all frames that contain it. It remains to prove that each frame is planar.
	
	\begin{itemize}
		
		\item Each frame $D_i$ with $i \in [\tau^-_1,\tau^+_1]$ is a subdrawing of $\Gamma[V]$, which is planar because it contains only the edges that are part of the common intersection of $G_1,G_2,G_3$. 
		
		\item Each frame $D_i$ with $i \in [\tau^-_2,\tau^+_3]$ is planar because it is a subdrawing of $\Gamma_1$ in which each edge has been subdivided and the spectators of such subdivision vertices have been drawn without introducing crossings. 
		
		\item Observe that in $\tau^-_4$, all spectators of the subdivision vertices of $E_1$ disappear and hence the same holds for the subdivision vertices of $E_1$. It follows that each frame $D_i$ with $i \in [\tau^-_4,\tau^+_5]$ is planar because it is a subdrawing of $\Gamma_2$ in which each edge has been subdivided and the spectators of such subdivision vertices have been drawn without introducing crossings.  The same argument (with respect to $E_3$) can be used to prove the planarity of each frame $D_i$ with $i \in [\tau^-_6,\tau^+_7]$.
		
		\item In $\tau^-_8$ all subdivision vertices of $E_3$ disappear and  again each frame $D_i$ with $i \in [\tau^-_8,\tau^+_8]$ is a subdrawing of $\Gamma[V]$ where some spectators of the vertices in $V$ have been drawn without introducing crossings. 
		
	\end{itemize}

	Suppose now that $\langle G, \tau \rangle$ admits a storyplan $\langle \tau, \{D_i\}_{i\in [n]}\rangle$. Then a SEFE $\Gamma_1,\Gamma_2,\Gamma_3$ of $\langle G_1, G_2, G_3 \rangle$ can be obtained as follows. Let $\Gamma^*_1 = \bigcup_{i=\tau^-_1}^{\tau^+_2} D_i$, let $\Gamma^*_2 = \bigcup_{i=\tau^-_4}^{\tau^+_4} D_i$, and let $\Gamma^*_3 = \bigcup_{i=\tau^-_6}^{\tau^+_6} D_i$. Observe that, for each $i=1,2,3$, $\Gamma^*_i$ is a planar drawing of $G_i$ in which the exclusive edges of $E_i$ are subdivided. Moreover, by construction it holds $\Gamma^*_1[V] = \Gamma^*_2[V] = \Gamma^*_3[V]$. It follows that $\Gamma_i$ can be obtained by viewing each subdivision vertex as an inner point of the corresponding edge. 
	
	The membership to \NP~follows the lines of \cref{th:upper}, simply ignoring the initial guess of a total order, which is instead given as part of the input.
\end{proof}

\end{document}